\documentclass[11pt]{article}
\usepackage{amsmath,amsfonts,amssymb,amsthm,latexsym, epsfig, graphics,mathrsfs, mathpazo,dsfont,wrapfig,bbm, fixmath, bm,multirow}
\usepackage{courier}
\usepackage[margin=1in,vmargin=1in]{geometry}
\usepackage{color}

\usepackage[numbers]{natbib}
   \theoremstyle{plain}
    \numberwithin{equation}{section}
    \newtheorem{thm}{Theorem}[section]
    \newtheorem{prop}[thm]{Proposition}

    \newtheorem{lem}[thm]{Lemma}
    \newtheorem{remk}{Remark}

\usepackage{xcolor}
 \definecolor{Red}{rgb}{0,0,0}
\newcommand{\Red}{\color{Red}}
\definecolor{Blue}{rgb}{0,0,0}

\definecolor{Green}{rgb}{0,0,0}

\definecolor{Magenta}{rgb}{0,0,0}

\begin{document}

\title{Optimal execution with liquidity risk in a diffusive order book market}

\author{ Hyoeun Lee\footnote{Department of Statistics, University of Illinois, Email:hyoeun@illinois.edu}
,\,\,Kiseop Lee\footnote{Department of Statistics, Purdue University, Email: kiseop@purdue.edu}
}
\maketitle

\begin{abstract}

We study the optimal order placement strategy with the presence of a liquidity cost. In this problem, a stock trader wishes to clear her large inventory by a predetermined time horizon $T$. A trader uses both limit and market orders, and a large market order faces an adverse price movement caused by the liquidity risk. First, we study a single period model where the trader places a limit order and/or a market order at the beginning. We show the behavior of optimal amount of market order, $m^*$, and optimal placement of limit order, $y^*$, under different market conditions. Next, we extend it to a multi-period model, where the trader makes sequential decisions of limit and market orders at multiple time points.
\end{abstract}

\section{Introduction}
%

\par In general, stock traders often need to handle a large order. Usually, the first step is to split a large order into multiple small orders before placing. This is to reduce the unfavorable price movements to the trader caused by a large order. Selling an asset tends to move the price downward, while buying an asset tends to move the price upward. This effect is often called as a liquidity cost, a price impact, or a market impact. In this paper, we are going to call this as a liquidity cost. 

\medskip
\par The liquidity cost affects the optimal strategy of traders, since a large initial market order may face a hefty liquidity cost. The optimal execution strategy under a liquidity cost has been studied extensively. 

  \par The pioneering work of \citet{bertsimas1998optimal}  
and  \citet{almgren2001optimal} consider a linear impact, such that the liquidity cost is proportional to the number shares of an order. \citet{bank2004hedging}, \citet{CJP:2003} , and \citet{frey2002risk} investigate deeply on the liquidity cost and additional transient impacts. 
 \citet{obizhaeva2013optimal} consider  a transient (linear) impact additionally. A transient impact is that the impact from a market order is not permanent, but only exists for a short moment and vanishes.
 
\par However, \citet{potters2003more},  \citet{eisler2012price}, and  \citet{donier2012market} empirical showed that the liquidity cost is not linear, but rather concave. \citet{alfonsi2010} ,    \citet{predoiu2011optimal},   \citet{gatheral2010no} and   \citet{gueant2015optimal}, and many others, have proposed extensions or alternatives to \citet{obizhaeva2013optimal} with a nonlinear liquidity cost.  

\par Furthermore, instead of the exponential decay of the transient impact, more general decay kernels are considered by  \citet{alfonsi2012order}  and \citet{gatheral2012transient} since the exponential decay is not observed on market data. 

\medskip

 \par When placing orders, the second step for investors is to decide between a market order and a limit order. A market order (MO) is an order to buy/sell an asset at the best available price. Market orders are executed immediately, while the best available price might be affected adversely (liquidity cost). Most studies so far have focused on an execution problem using solely a market order. 
 \par A limit order (LO) is an order to buy/sell an asset at a specific price. The price of a limit order is specified by the buyer/seller, but the execution is not guaranteed. A limit order can be cancelled by the buyer/seller before execution.

\par A limit order book (LOB) collects all quantities and the price of limit orders. The LOB is updated upon execution of market orders, submission of limit orders, or cancellation of pre-existing limit orders. 
\par Recently, more researchers like \citet{jacquierLiu2017}, \citet{alfonsi2010}, \citet{GuilbaudPham2013}, \citet{Manglaras2015}, \citet{ContKukanov2013} consider both market orders and limit orders, instead of focusing on the market order only. However, only the best bid or the best ask prices are considered as options.

\citet{cartea2014modelling}, \citet{GDR:2017}, \citet{FLP:2018}
have considered whether placing the limit order deeper in the book could be preferable, which is often called as an optimal placement problem. In \cite{cartea2014modelling}, the optimal placement problem is studied under a continuous-time model for the mid-price, a mean-reverting process with jumps. \cite{GDR:2017} investigate the optimal placement problem under a discrete-time model for the level I prices of a LOB. In \cite{FLP:2018}, a problem similar to \cite{GDR:2017} is studied under a continuous-time model, and shows that there exists an optimal placement policy different from the Level I-II solution of \cite{GDR:2017}.  

While  \cite{GDR:2017} and \cite{FLP:2018} studied the optimal placement behavior with a small size order without a liquidity cost, practically we cannot entirely ignore the liquidity cost. Therefore, in this paper, we solve the optimal execution problem with liquidity cost, while considering both market orders and limit orders as options. For the limit order, the optimal price level is also investigated as in \cite{FLP:2018}. For the price impact model, we borrow the  liquidity risk introduced in \cite{CJP:2003} but we add the necessary features of LOB. 

\par We first investigate the optimal placement problem under the single-period setting, where the investor can place an order only at the beginning. The investor needs to clear the inventory (of size $M$) by a certain time horizon $T$. At $t=0$, the investor decides the quantity of the market order $m$ ($0\leq m \leq M$), and the rest $M-m$ will be placed using the limit order. The price level for the limit order, $y$, is also determined at $t=0$. At time $T$, any unexecuted limit order is immediately executed using a market order. 

\par Next, we study an analogous problem with a multi-period setting, where the investor makes decisions at multiple time steps, $\{0,T/n, 2T/n , \dots T \}$. At each time step, the investor cancels existing un-executed limit orders and place new market and limit orders. At time $T$, similar to single-period setting, any remaining inventory is immediately executed using a market order.

\par We model the asset price using the Brownian and the geometric Brownian motion, and also model the liquidity cost from the investor's market order; we use the model derived from \cite{CJP:2003}, but also include the basic feature of Limit Order book. We investigate the behavior of optimal $m^*$ and $y^*$ which maximize the expected cash flow from the order placement. 
 
\par The rest of the paper is organized as follows. We briefly explain the liquidity risk introduced in \cite{CJP:2003} in section 2. In section 3, we discuss the optimal execution strategy in one period model. We extend it to a multi-period case in section 4. We conclude in section 5.


\section{Liquidity Risk}

We recall the concepts introduced in the
work of \citet{CJP:2003}. We consider a market with a risky asset and
a money market account. The risky asset, stock, pays no dividend and we assume
that the spot rate of interest is zero, without loss of generality.
$S(t,x,\omega)$ represents the stock price per share at time
$t \in [0,T]$ that the trader pays/receives for an order of size $x \in \mathbf{R}$ given the state $\omega \in \Omega$.
A positive order($x>0$) represents a buy, a negative order ($x<0$)
represents a sale, and the zeroth order ($x=0$) corresponds to
the marginal trade. For the detailed structure of the supply curve, we refer to Section 2 of \c{C}etin et al.~\cite{CJP:2003}.

A \emph{trading strategy} (portfolio) is a triplet $((X_t,Y_t : t \in [0,T]),\tau)$ where $X_t$ represents the trader's aggregate stock holding at time $t$
(units of the stock), $Y_t$ represents the trader's aggregate money account position at time $t$ (units of money market account), and $\tau$
represents the liquidation time of the stock position. Here, $X_t$ and $Y_t$ are predictable and optional processes, respectively,
with $X_{0-} \equiv Y_{0-} \equiv 0$.  A \emph{self-financing strategy} is a trading strategy $((X_t,Y_t : t \in [0,T]),\tau)$ where
 $X_t$ is cadlag if $\frac{\partial S}{\partial x}=0$ for all $t$, and $X_t$ is cadlag with finite quadratic variation
       ($[X,X]_T < \infty$) otherwise, and
\begin{align}
 Y_t &= Y_0 + X_0 S(0,X_0) + \int_0^t X_u dS(u,0)-X_t S(t,0) \notag \\
     &- \sum_{0\leq u \leq t} \Delta X_u [S(u, \Delta X_u)-S(u,0)] - \int_0^t \frac{\partial S}{\partial x}(u,0) d[X,X]_u^c. \label{SF}
\end{align}

Therefore, it is natural to define \emph{the liquidity cost} of a self-financing trading strategy $(X, Y, \tau)$ by
$L_t = \sum_{0\leq u \leq t} \Delta X_u [S(u, \Delta X_u)-S(u,0)] + \int_0^t \frac{\partial S}{\partial x}(u,0) d[X,X]_u^c, \quad 0<t \leq T,$
where $L_{0-} = 0$, and $L_0 = X_0[S(0,X_0)-S(0,0)]$. 

A practically important problem of a trader is how to execute a large order in the market with liquidity risk. Guo et al. \cite{GDR:2017} studied an optimal placement problem under a discrete time setting. Figueroa-Lopez et al. \cite{FLP:2018} studied an optimal placement problem in continuous time setting, but both studies did not consider the liquidity cost. We plan to study how to place an order using both market and limit orders in a market with liquidity risk.

\section{A Single Period Model}

\subsection{A Diffusive LOB Market}

As the first step, let us consider a simple case where the marginal price $S(t,0)$ is a diffusion process. The goal is to sell $M$ orders by time horizon $T$. We consider two cases. In a single period case, we will study the optimal placement problem in a single step (decision is made once when $t=0$). The trader will use market order and limit order to minimize the utility function (expected cost), which incorporates the execution risk from limit order and liquidity cost from market order. We next extend it to a multi-period case. Then the goal is to use the previous step to build the multiple-period (decision can be made multiple times ($t=0,t_1,\dots,t_N=T$), and at each time step the trader place market/limit order.

For the supply curve, we use
\[
S(t,x)=S(t,0)-d-\beta(K+x)^{-}, \,\, \beta>0,\,\,x<0,
\]
and 
\[
S(t,0)=S_0+\mu t+\sigma W_t,  \quad \text{or} \quad S(t,0)=S_0 e^{(\mu-\frac{\sigma^2 }{2}) t +\sigma W_t},
\]
where $d(>0)$ is the half of the bid and ask spread, and $K$ is the initial market depth at the best bid price. Therefore, the price does not move up to the first $K$ shares, then moves down linearly afterward.{ \cite{BP:2010}, \cite{CJP:2003} suggest linear supply curve: For liquid stocks, the supply curve is linear in $x$ and for highly illiquid commodities, \cite{BP:2010} suggests a jump-linear supply curve. In our work, we apply a linear curve, and just use one $\beta$ since we focus on only one side (sell). Also, we consider the necessary limit order book feature, the market depth $K$.

First, let us consider a single period case. The placement is made only once at time $0$. At time $0$, the trader makes a decision. If a trader sells $M$ orders using only a market order at time 0, the trader's cash flow becomes $M*(S(0,-M)-f )$ where $f$ is the fee per an executed market order. A more general case is when a trader sells $m$ orders using a market order at time 0, and places a $M-m$ sell limit order at the price level $S(0,0)+y$, $y\geq d$. Any remaining unexecuted orders at time $T$ are converted to a market order and  are executed immediately with paying the liquidity cost.

Let $\tau$ be the first time that the trader's limit order becomes the best ask. In other word, $\tau$ is the first time $S(0,0)+y$ becomes the best ask. Let $L$ be the number of executed shares at the price level $S(0,0)+y$ and denote $L=(M-m)\rho$ where $\rho$ is the proportion of the execution.

Notice that even when the trader's limit order becomes the best ask, there is no guarantee of execution, since there will be orders at the exact same level from other traders too.
The trader's cash flow from the initial market order is $m*(S(0, -m)-f)$. If $\tau>T$, no limit order will be executed and all will be put as a market order. Therefore, the expected cash flow from the remaining $M-m$ shares will be $(M-m)E(S(T,-(M-m))-f|\tau>T).$
When $\tau<T$, the expected cash flow from the remaining $M-m$ shares becomes $E(L(S({0},0)+y+r))+E((M-m-L)(S(T,-(M-m))-f)   |\tau<T), $
where $r$ is the rebate per an executed limit order.

Our goal is to find the optimal $(m,y)$ which maximizes the expected cash flow
\begin{align}\nonumber
&m(S(0, -m)-f) +(M-m)E(S(T,-(M-m))-f|\tau>T)P(\tau>T) \\
&+\{E(L (S(0,0)+y+r))+E ((M-m-L)(S(T,-(M-m-L))-f)   |\tau<T)\}P(\tau<T).
\label{ECF_lem3}
\end{align}

 To do this, we need to calculate the conditional expectation $E(S(T,0)|\tau<T)$. Since $S(t,0)$ is a function of a Brownian motion, the distribution of $\tau$ is related to the hitting time of a Brownian motion. It is well known that it is obtained by the reflection property of a Brownian motion and its running maximum. That is a common problem especially in a barrier option. It is also known that when $\tau_a$ is the hitting time of $a$ for a standard Brownian motion $B_t$, $0< b \leq a$, and $M_t$ is its running maximum, we have
\[
P(B_T<b|\tau_a <T)= P(B_T > a +(a-b)|\tau_a <T),
\]
and
\[
P(M_T \geq a, \, B_T <b)=P(M_T \geq a, \, B_T>2a-b).
\]
We will apply this joint distribution to calculate the conditional expectation $E(S(T,0)|\tau<T)$. It becomes a function of $y$ only. 

{
Next two lemmas give us useful equations about expected stock prices and the probability of the limit order execution. Lemma~\ref{BM_Lem} describes the expected stock price, which follows a Brownian motion, in two different cases: when the stock price hits a certain price before $t$ and when hitting does not happen until $t$. Lemma~\ref{GBM_Lem} shows analogous expected stock prices when the stock price follows a geometric Brownian motion.
}
{\Red
\begin{remk}
For the rest of the paper, the pdf, cdf, and survival or tail distribution of a standard normal r.v. $Z$ are denoted by $\phi(z)=e^{-z^{2}/2}/\sqrt{2\pi}$, $N(z)=\int_{-\infty}^{z}\phi(x)dx$. $B_t$ is a standard Brownian motion. 
\end{remk}

}

{
\begin{lem}\label{BM_Lem}
	Let us assume that the price process follows a Brownian motion with drift, $S(t,0)=S(0,0)+\mu t +\sigma B_t$. Let $\tau_y:=\inf_{u} \{S(u,0)+d=S(0,0)+y \}.$ Then
	
	\begin{align}\nonumber
		P(\tau_y>t)=&N\left(\frac{(y-d)-\mu t}{\sigma\sqrt{t}}\right)-e^{\frac{2(y-d)\mu }{\sigma^2} } N\left(\frac{-(y-d)-\mu t}{\sigma\sqrt{t}}\right),\\ \nonumber
		E[(S(t,0)-S(0,0)) \mathcal{I}(\tau_y>t)]=&\mu t N\left(\frac{y-d-\mu t}{\sigma\sqrt{t}}\right) +e^{\frac{2(y-d)\mu }{\sigma^2} }(-2(y-d)-\mu t)  N\left(\frac{-(y-d)-\mu t}{\sigma\sqrt{t}}\right),\\ \nonumber
		E[(S(t,0)-S(0,0))\mathcal{I}(\tau_y<t)]=&\mu t N\left(\frac{-(y-d)+\mu t}{\sigma\sqrt{t}}\right) +e^{\frac{2(y-d)\mu }{\sigma^2} }(2(y-d)+\mu t)  N\left(\frac{-(y-d)-\mu t}{\sigma\sqrt{t}}\right).
	\end{align}

\end{lem}

\begin{proof}
	{\Red  Let $M_t:=\max\limits_{u\leq t} S(u,0)-S(0,0)$. From the definition of $\tau_y$, $\{ \tau_y>t\}= \{ M_t \leq y-d\}$.
	From \citet{jeanblanc2009mathematical}, note that 
	$$P((S(t,0)-S(0,0))\in dz, M_t\leq a)= \phi\left(\frac{z-\mu t}{\sigma\sqrt{t}}\right)\frac{1}{\sigma\sqrt{t}} -e^{2\mu a}\phi\left(\frac{z-2a-\mu t}{\sigma\sqrt{t}}\right).$$
		Then
		\begin{align*}
			P(\tau_y>t)= P( M_t \leq y-d) &= 
		\int_{-\infty}^{(y-d)}
		\phi\left(\frac{z-\mu t}{\sigma\sqrt{t}}\right)\frac{1}{\sigma\sqrt{t}} -e^{2\mu (y-d)}\phi\left(\frac{z-2(y-d)-\mu t}{\sigma\sqrt{t}}\right) dz
		\\
		&=N\left(\frac{(y-d)-\mu t}{\sigma\sqrt{t}}\right)-e^{\frac{2(y-d)\mu }{\sigma^2} } N\left(\frac{-(y-d)-\mu t}{\sigma\sqrt{t}}\right).
		\end{align*} 
		
		Next, 
	\begin{align*}
		E[(S(t,0)&-S(0,0))\mathcal{I}(\tau_y>t)]=E[(S(t,0)-S(0,0))\mathcal{I}(M_t\leq (y-d))]
		\\
		=&\int_{-\infty}^{(y-d)} z P((S(t,0)-S(0,0))\in dz, M_t\leq (y-d))\\
		=& \int_{-\infty}^{(y-d)} z
		\phi\left(\frac{z-\mu t}{\sigma\sqrt{t}}\right)\frac{1}{\sigma\sqrt{t}} -e^{2\mu (y-d)}\phi\left(\frac{z-2(y-d)-\mu t}{\sigma\sqrt{t}}\right) dz
		\\
		=&\mu t N\left(\frac{y-d-\mu t}{\sigma\sqrt{t}}\right) +e^{\frac{2(y-d)\mu }{\sigma^2} }(-2(y-d)-\mu t)  N\left(\frac{-(y-d)-\mu t}{\sigma\sqrt{t}}\right). 
	\end{align*}
	
	Finally, 
	\begin{align*}
			E[(S(t,0)-S(0,0))\mathcal{I}(\tau_y<t)]&=E[(S(t,0)-S(0,0))]-E[(S(t,0)-S(0,0))\mathcal{I}(\tau_y>t)]\\
	&=\mu t - E[(S(t,0)-S(0,0))\mathcal{I}(\tau_y>t)]\\
	&= \mu t N\left(\frac{-(y-d)+\mu t}{\sigma\sqrt{t}}\right) +e^{\frac{2(y-d)\mu }{\sigma^2} }(2(y-d)+\mu t)  N\left(\frac{-(y-d)-\mu t}{\sigma\sqrt{t}}\right).
	\end{align*}
	
	}

%
\end{proof}



The next lemma gives us similar results when the price process is a geometric Brownian motion.

\begin{lem}\label{GBM_Lem}
	Let us assume that the price process follows a geometric Brownian motion, $dS(t,0)=\mu S(t,0) dt + \sigma S(t,0) dB_t $. Also, let's denote $\tau_y:=\inf_{u} \{S(u,0)=S(0,0)+y-d \}$. Then 
	{\Red 
	
	\begin{align*}
		P(\tau_y>t)=&N\left( \frac{a-\mu t+\sigma^2 /2}{\sigma\sqrt{t}} \right)- e^{\frac{2a\mu}{\sigma^2}-a}N\left( \frac{-a-\mu t+\sigma^2 t/2}{\sigma\sqrt{t}} \right),\\
		E[S(t,0) \mathcal{I}(\tau_y>t)]=&S(0,0) \left( e^{\mu t} N\left( \frac{a-\mu t +\sigma^2t/2 }{\sigma\sqrt{t}} \right)-e^{\frac{2a\mu}{\sigma^2}+\mu t + a}N\left(\frac{-a-\mu t +\sigma^2 t/2}{\sigma\sqrt{t}} \right)  \right),\\
		E[S(t,0) \mathcal{I}(\tau_y<t)]=&S(0,0) \left( e^{\mu t} N\left( \frac{-a+\mu t -\sigma^2t/2 }{\sigma\sqrt{t}} \right)+e^{\frac{2a\mu}{\sigma^2}+\mu t + a}N\left(\frac{-a-\mu t +\sigma^2 t/2}{\sigma\sqrt{t}} \right)  \right) .
	\end{align*}
	
	}
	where
	$$a:=\ln \left(\frac{(S(0,0)+y-d)}{S(0,0)}\right).$$
\end{lem}

\begin{proof}
{\Red 
	Let $X_t=\ln (S(t,0)/S(0,0))$. Then $X_t=(\mu-\sigma^2/2) t+\sigma B_t$. 
	Let's denote $M_t:=\max\limits_{u\leq t} X_u$. From the definition of $\tau_y$, note that  $\{\tau_y>t \} =\{ {M_t} \leq \ln \left(\frac{(S(0,0)+y-d)}{S(0,0)}\right)\}$.  
	Then, we may apply the proof of Lemma~\ref{BM_Lem} by using $a=\ln \left(\frac{(S(0,0)+y-d)}{S(0,0)}\right)$ and $P(X_{t}\in dz, M_t\leq a)=\phi\left(\frac{z-(\mu-\frac{\sigma^2}{2}) t}{\sigma\sqrt{t}}\right)\frac{1}{\sigma\sqrt{t}} -e^{2(\mu-\frac{\sigma^2}{2}) a}\phi\left(\frac{z-2a-(\mu-\frac{\sigma^2}{2})t}{\sigma\sqrt{t}}\right).$
	
	First, 
	\begin{align*}
			P(\tau_y>t)= P( M_t \leq a) &= 
		\int_{-\infty}^{a}
		\phi\left(\frac{z-(\mu-\frac{\sigma^2}{2}) t}{\sigma\sqrt{t}}\right)\frac{1}{\sigma\sqrt{t}} -e^{2(\mu-\frac{\sigma^2}{2}) a}\phi\left(\frac{z-2a-(\mu-\frac{\sigma^2}{2})t}{\sigma\sqrt{t}}\right) dz
		\\
		&=N\left( \frac{a-\mu t+\sigma^2 /2}{\sigma\sqrt{t}} \right)- e^{\frac{2a\mu}{\sigma^2}-a}N\left( \frac{-a-\mu t+\sigma^2 t/2}{\sigma\sqrt{t}} \right)
		\end{align*} 

Next, 
\begin{align*}
	E\left[\frac{S(t,0)}{S(0,0)} \mathcal{I}(\tau_y>t)\right]
	&=E[e^{X_t}\mathcal{I}(M_t\leq a)]=\int_{-\infty}^{a} e^z P(X_t\in dz, M_t\leq a)\\
	&= \left( e^{\mu t} N\left( \frac{a-\mu t +\sigma^2t/2 }{\sigma\sqrt{t}} \right)-e^{\frac{2a\mu}{\sigma^2}+\mu t + a}N\left(\frac{-a-\mu t +\sigma^2 t/2}{\sigma\sqrt{t}} \right)  \right),\\
	E\left[\frac{S(t,0)}{S(0,0)} \mathcal{I}(\tau_y<t)\right]
	&=E[e^{X_t}] - E[e^{X_t}\mathcal{I}(M_t\leq a)]\\
	&= e^{\mu t}- \left( e^{\mu t} N\left( \frac{a-\mu t +\sigma^2t/2 }{\sigma\sqrt{t}} \right)-e^{\frac{2a\mu}{\sigma^2}+\mu t + a}N\left(\frac{-a-\mu t +\sigma^2 t/2}{\sigma\sqrt{t}} \right)  \right)\\
	&= \left( e^{\mu t} N\left( \frac{-a+\mu t -\sigma^2t/2 }{\sigma\sqrt{t}} \right)+e^{\frac{2a\mu}{\sigma^2}+\mu t + a}N\left(\frac{-a-\mu t +\sigma^2 t/2}{\sigma\sqrt{t}} \right)  \right).
\end{align*}

}
\end{proof}


{
Using Lemma~\ref{BM_Lem} , we can calculate the expected cash flow (ECF) for a Brownian motion model. Similarly, the ECF for a geometric Brownian model is explicitly shown in Lemma~\ref{ECF_GBM} using results from Lemma~\ref{GBM_Lem}. 
}
{\Red
\begin{remk}\label{Remk:notation_alphas}
	For the rest of the paper, we use the following notations:
	\[
\alpha_t= \frac{y-d+\mu t}{\sigma \sqrt{t}},
\quad
\beta_t = \frac{y-d-\mu t}{\sigma\sqrt{t}},
\quad \tilde{N}(-\alpha_t)= e^{2(y-d)\mu/\sigma^2}N(-\alpha_t),
\quad
\epsilon_0= \left(N\left(\beta_{T} \right)-\tilde{N}\left(-\alpha_{T}\right) \right).
\]
\end{remk}
}

{\Red 
\begin{lem}\label{ECF_BM}
	 Let us assume that the price process follows a Brownian motion, $S(t,0)=S(0,0)+\mu t +\sigma B_t$. Then $ECF(y,m)$ can be summarized as follows:
 
 \begin{equation}
 	\begin{aligned}
ECF(y,m)&=MS(0,0)+m (-d-\beta (K-m)^{-}-f) \\
&+(M-m)\left\{\mu T N(\beta_T) -(2(y-d)+\mu T) \tilde{N} (-\alpha_T)+\epsilon_0 (-d-\beta(K-M+m)^{-}-f)\right\} \\
&+\left(1-\epsilon_0\right)\left(E[L](y+r)-E[(M-m-L)(d+\beta(K-M+m+L)^{-}+f)] \right)\\
&+(M-m-E[L])\left(\mu T N\left(-\beta_T \right) +(2(y-d)+\mu T)\tilde{N}(-\alpha_T)\right).
\end{aligned}
\label{BM_ECF_eq}
 \end{equation}


\end{lem}

\begin{proof}

By definition of $S(t,-m)$, $ECF(y,m)$ can be rewritten as follows:

\begin{equation}
	\begin{aligned}
		ECF(y,m)&=m (S_0-d-\beta (K-m)^{-}-f) \\
&+(M-m)\left\{E[S(T,0)\mathcal{I}(\tau>T)]+P(\tau>T) (-d-\beta(K-M+m)^{-}-f)\right\} \\
&+E[L](S(0,0)+y+r)P(\tau<T)-E[(M-m-L)(d+\beta(K-M+m+L)^{-}+f)] P(\tau<T)\\
&+E[(M-m-L)]E[(S(T,0)\mathcal{I} (\tau<T)] .
	\end{aligned}
	\label{BM_ECF:prf_exp1}
\end{equation}

	

From Lemma~\ref{BM_Lem}, note that 
\begin{equation}
	 \begin{aligned}
 	E[S(t,0) \mathcal{I}(\tau >t)] &=S(0,0) P(\tau >t)+ \mu t N\left(\beta_t \right) -(2(y-d)+\mu t)\tilde{N}(-\alpha_t), \\
 	E[S(t,0) \mathcal{I}(\tau < t)] &=S(0,0) P(\tau <t)+\mu t N\left(-\beta_t \right) +(2(y-d)+\mu t)\tilde{N}(-\alpha_t),\\ 
 	P(\tau >t) &=N(\beta_t)-\tilde{N}(-\alpha_t). 
 \end{aligned}
 	\label{BM_ECF:prf_exp2}
\end{equation}

Then, by plugging expressions in (\ref{BM_ECF:prf_exp2}) into (\ref{BM_ECF:prf_exp1}), we have the final expression as in (\ref{BM_ECF_eq}).

	
\end{proof}
}

Similarly, using Lemma 2, we can calculate the expected cash flow for a geometric Brownian motion model.

{\Red

\begin{lem}\label{ECF_GBM}
 Let us assume that the price process follows geometric Brownian motion, $dS_t=\mu S_t dt + \sigma S_t dB_t $. Then $ECF(y,m)$ can be summarized as follows:
\begin{equation}
 \begin{aligned}
ECF(y,m)&=m (S_0-d-\beta (K-m)^{-}-f) \\
&+(M-m)\left\{S_0 \left( e^{\mu T} N\left( \tilde{\beta}_T \right)-\tilde{N}\left(-\tilde{\alpha}_T \right)  \right)+(N\left( \tilde{\beta}_T \right)- \tilde{N}\left( -\tilde{\alpha}_T \right)
) (-d-\beta(K-M+m)^{-}-f)\right\} \\
&+E[L](S(0,0)+y+r)\left(N\left( -\tilde{\beta}_T \right)+ \tilde{N}\left( -\tilde{\alpha}_T \right)
\right)\\
&-E[(M-m-L)(d+\beta(K-M+m+L)^{-}+f)] \left(N\left( -\tilde{\beta}_T \right)+ \tilde{N}\left( -\tilde{\alpha}_T \right)
\right)\\
&+E[(M-m-L)]S_0 \left( e^{\mu T} N\left( -\tilde{\beta}_T \right)+\tilde{N}\left(-\tilde{\alpha}_T \right) \right),
\end{aligned}	
\label{ECF_GBM:lemEq}
\end{equation}

where 	
\[
a:=\ln \left(\frac{(S_0+y-d)}{S_0}\right),\quad 
\tilde{\alpha}_t= \frac{a+\mu t -\sigma^2t/2 }{\sigma\sqrt{t}},
\quad
\tilde{\beta}_t = \frac{a-\mu t +\sigma^2t/2 }{\sigma\sqrt{t}},
\quad \tilde{N}(-\tilde{\alpha}_t)= e^{\frac{2a\mu}{\sigma^2}+\mu t + a}N(-\tilde{\alpha}_t).
\]
 

\end{lem}

\begin{proof}
By definition of $S(t,-m)$, $ECF(y,m)$ can be rewritten as follows:
\begin{equation}
	\begin{aligned}
ECF(y,m)&=m (S_0-d-\beta (K-m)^{-}-f) \\
&+(M-m)\left\{E[S(T,0)\mathcal{I}(\tau>T)]+P(\tau>T) (-d-\beta(K-M+m)^{-}-f)\right\} \\
&+E[L](S(0,0)+y+r)P(\tau<T)-E[(M-m-L)(d+\beta(K-M+m+L)^{-}+f)] P(\tau<T)\\
&+E[(M-m-L)]E[(S(T,0)\mathcal{I} (\tau<T)] .
\end{aligned}
\label{ECF_GBM_prf_expression1}
\end{equation}

	

From Lemma~\ref{GBM_Lem}, note that 

\begin{equation}
	 \begin{aligned}
 	E[S(T,0) \mathcal{I}(\tau >T)] &=S_0 \left( e^{\mu T} N\left( \tilde{\beta}_T \right)-\tilde{N}\left(-\tilde{\alpha}_T \right)  \right), \\ 
 	E[S(T,0) \mathcal{I}(\tau < T)] &=S_0 \left( e^{\mu T} N\left( -\tilde{\beta}_T \right)+\tilde{N}\left(-\tilde{\alpha}_T \right) \right)  ,\\ 
 	P(\tau >T)&=N\left( \tilde{\beta}_T \right)- \tilde{N}\left( -\tilde{\alpha}_T \right),
 \end{aligned}
 \label{ECF_GBM_prf_expression2}
\end{equation}

Then, by plugging expressions in  (\ref{ECF_GBM_prf_expression2}) into (\ref{ECF_GBM_prf_expression1}), we have the final expression as in (\ref{ECF_GBM:lemEq}).

	
\end{proof}

}

\subsection{case analysis: a constant $\rho$ with a Brownian motion Model}
Let us recall that { $L=(M-m)\rho$ }, $\rho\in [0,1]$, where $\rho$ is the proportion of the execution such that $L=(M-m)\rho$. In this subsection, we work on the simple case when $\rho$ is a constant. 
{

\begin{remk}
	In practical situation, a constant  $\rho$ is an unrealistic assumption. A more realistic case for $\rho$ is a random model, such as $\rho=\frac{1}{2} + \frac{2}{\pi}\arctan(X)$ where $X$ follows a normal distribution, which could depend on $y$ and $m$. Or, $\rho$ could be a function of flows and time where we understand arrivals of orders as some counting process. However, for simplicity, we consider $\rho$ as a constant in this paper. \\
\end{remk}

}

In the following Lemma, we introduce the ECF (Expected Cash Flow) function when the price process follows a Brownian motion. 

\begin{lem}\label{lem:bm:ECF}
Let us assume that the price process follows Brownian motion, $S(t,0)=S(0,0)+\mu t +\sigma B_t$, where $B_t$ is a standard Brownian motion. 
When  { $L=(M-m)\rho$ }, $ECF(y,m)$ in   (\ref{BM_ECF_eq}) can be expressed as follows:

\begin{align*}
ECF(y,m)=
&M \left(S_0 -2\rho(y-d)\tilde{N}(-\alpha_{T})-(d+f-\mu {T})(1-\rho+\rho\epsilon_0)+(1-\epsilon_0)\rho(y+r)\right)\\
&-m\left( \rho(1-\epsilon_0)(d+f+y+r)-2\rho (y-d)\tilde{N}(-\alpha_{T}) +\mu {T}(1-\rho+\rho\epsilon_0) \right)\\
&-\beta\left[m (K-m)^{-}+(M-m)\left(\epsilon_0  (K-M+m)^{-}+(1-\epsilon_0)(1-\rho)(K-(M-m)(1-\rho))^{-}\right) \right],
\end{align*}
%
where the notations  $\alpha_T$, $\beta_T$, $\epsilon_0$ are from Remark~\ref{Remk:notation_alphas}.

\end{lem}

\begin{proof}
	{\Red This result is from the equation (\ref{BM_ECF_eq}) of  Lemma~\ref{ECF_BM} with $L=(M-m)\rho$.
	}
\end{proof}

{
Using this Lemma~\ref{lem:bm:ECF}, we now study the optimal placement strategy.
Next two theorems, Theorem~\ref{thm:single:muNeg} and Theorem~\ref{thm:single:muPos} give us optimal strategies when $\mu<0$ and $\mu>0$, respectively. We find the optimal $(y^*,m^*)$, which means that the optimal strategy at time $0$ is to place a market order of quantity $m^*$, and place remaining $M-m^*$ orders by a limit order at the price of $S(0,0)+y^*$. 
}


\begin{thm}\label{thm:single:muNeg}
	Let $(y^*, m^*)$ be the optimal strategy satisfying

$$ECF(y^*,m^*)\geq ECF(y,m) \quad \forall y>0,\quad  0\leq m\leq M. $$

{\Red Under the cash flow model from  Lemma~\ref{lem:bm:ECF}, let $\mu<0$. Then $y^*=d$. The behavior of $m^*$ is as follows:}
\begin{enumerate}
	\item If $\rho=1$, $m^*=0$.
	\item If $\rho<1$ and $K\geq M$,
		\begin{itemize}
			\item $m^*=0$ when $\rho(2d+r+f)+\mu t(1-\rho)>0$
			\item $m^*=M$ when $\rho(2d+r+f)+\mu t(1-\rho)<0$
		\end{itemize}
	\item If $\rho<1$ and $K< M$,	
	\item[] $m^*$ is one of $\{0,K, M-K/(1-\rho), M, m^*_1, m^*_2, m^*_3\}$ where  
\end{enumerate}
$$ 
m^*_1=K/2-\epsilon_3/2\beta, \quad
m^*_2=M-\frac{K}{2(1-\rho)}-\frac{\epsilon_3}{2(1-\rho)^2\beta},
\quad 
m^*_3=\frac{2M\beta (1-\rho)^2 +\beta\rho K-\epsilon_3}{2\beta(1+(1-\rho)^2)},
$$ 

$$\epsilon_3:=\rho(2d+r+f)+\mu t(1-\rho).
$$

%
%

\end{thm}

}

\begin{proof}
\noindent\textbf{Part 1: Behavior of $y^*$: }
	Note that
	
\begin{align}
\label{ECF_rho_der_y.1}
\frac{\frac{\partial ECF}{\partial y}}{ (M-m)} &=
\rho(N(-\beta_t)-\tilde{N}(-\alpha_t))\\
\label{ECF_rho_der_y.2}
&+\frac{2e^{\frac{2(y-d)\mu}{\sigma^2}} }{\sigma\sqrt{t}}
\left( \phi (\alpha_t ) -\alpha_t N(-\alpha_t)\right)
\left(-\rho (2d+f+r-\mu t)+ \beta \gamma \right) \\
\label{ECF_rho_der_y.3}
&+\tilde{N}(-\alpha_t)\frac{2(y-d)}{\sigma^2 t}\left(-\rho (2d+f+r)+ \beta \gamma \right) 
\end{align}

$$\gamma:=\min(K-M+m,0)-(1-\rho)\min (K-(M-m)(1-\rho),0)\quad \quad  (note: \gamma<0)
$$
When $\mu<0$, all three lines (\ref{ECF_rho_der_y.1}), (\ref{ECF_rho_der_y.2}), (\ref{ECF_rho_der_y.3}) are negative, so $y^*=d$.\\

\noindent\textbf{Part 2: Behavior of $m^*$: }
 We've already shown that when $\mu<0$, $y^*=d$. Therefore,  $ECF(y=y^*, m)$ can be simplified as following:

\begin{gather}
	\begin{aligned}
	ECF(y=d,m)=&M (S_0+\rho(d+r)-(1-\rho)(d+f-\mu t))-m\epsilon_3 +m\beta \min (K-m,0)\\
	&+\beta (M-m)(1-\rho) \min(K-(M-m)(1-\rho),0),
\end{aligned}\label{ecf:thm1p2c2}
\end{gather}

where $\epsilon_3:=\rho(2d+r+f)+\mu t(1-\rho)$. We will divide the range of $m$ to four cases to obtain the expression of $ECF(y,m)$ and to obtain $m^*$ which maximizes $ECF$. 

\begin{enumerate}
	\item $K-(M-m)(1-\rho)\geq 0$, $K-m\geq 0$
	\item[] In this case, (\ref{ecf:thm1p2c2}) becomes a decreasing function of $m$ if $\epsilon_3>0$ and increasing function of $m$ if $\epsilon_3>0$. Note that if $K\geq M$, for any $m$ the conditions ($K-(M-m)(1-\rho)\geq 0$, $K-m\geq 0$) are satisfied. Therefore, if $K\geq M$, $m^*=0$ if $\epsilon_3>0$ and $m^*=M$ if $\epsilon_3<0$. This includes the case when $\rho=1$: when $\rho=1$, $\epsilon_3=2d+r+f>0$, so $m^*=0$.
	\item  $K-(M-m)(1-\rho)\geq 0$ (i.e. $m\geq M-K/(1-\rho)$), $K-m< 0$
	\item[] In this case, maximum of (\ref{ecf:thm1p2c2}) is obtained at $m_1^*=K/2-\epsilon_3/2\beta$. 
	\item $K-(M-m)(1-\rho)< 0$ (i.e. $m< M-K/(1-\rho)$), $K-m\geq 0$
	\item[] In this case, maximum of (\ref{ecf:thm1p2c2}) is obtained at
	$m_2^*=M-\frac{K}{2(1-\rho)}-\frac{\epsilon_3}{2(1-\rho)^2\beta}$. 
	\item $K-(M-m)(1-\rho)< 0$ (i.e. $m< M-K/(1-\rho)$), $K-m< 0$
	\item[] In this case, maximum of (\ref{ecf:thm1p2c2}) is obtained at
	$m_3^*=\frac{2M\beta (1-\rho)^2 +\beta \rho K-\epsilon_3}{2\beta(1+(1-\rho)^2)}$. 
\end{enumerate}

To summarize, when $K\geq M$, $m^*=0$ or $M$ depending on the sign of $\epsilon_3$. If not, the maximum will be obtained at one of $\{0, m_1^*, m_2^*, m_3^*, M, K, M-K/(1-\rho)\}$. 

\end{proof}

\begin{figure}
	\includegraphics[width=0.495\textwidth]
	{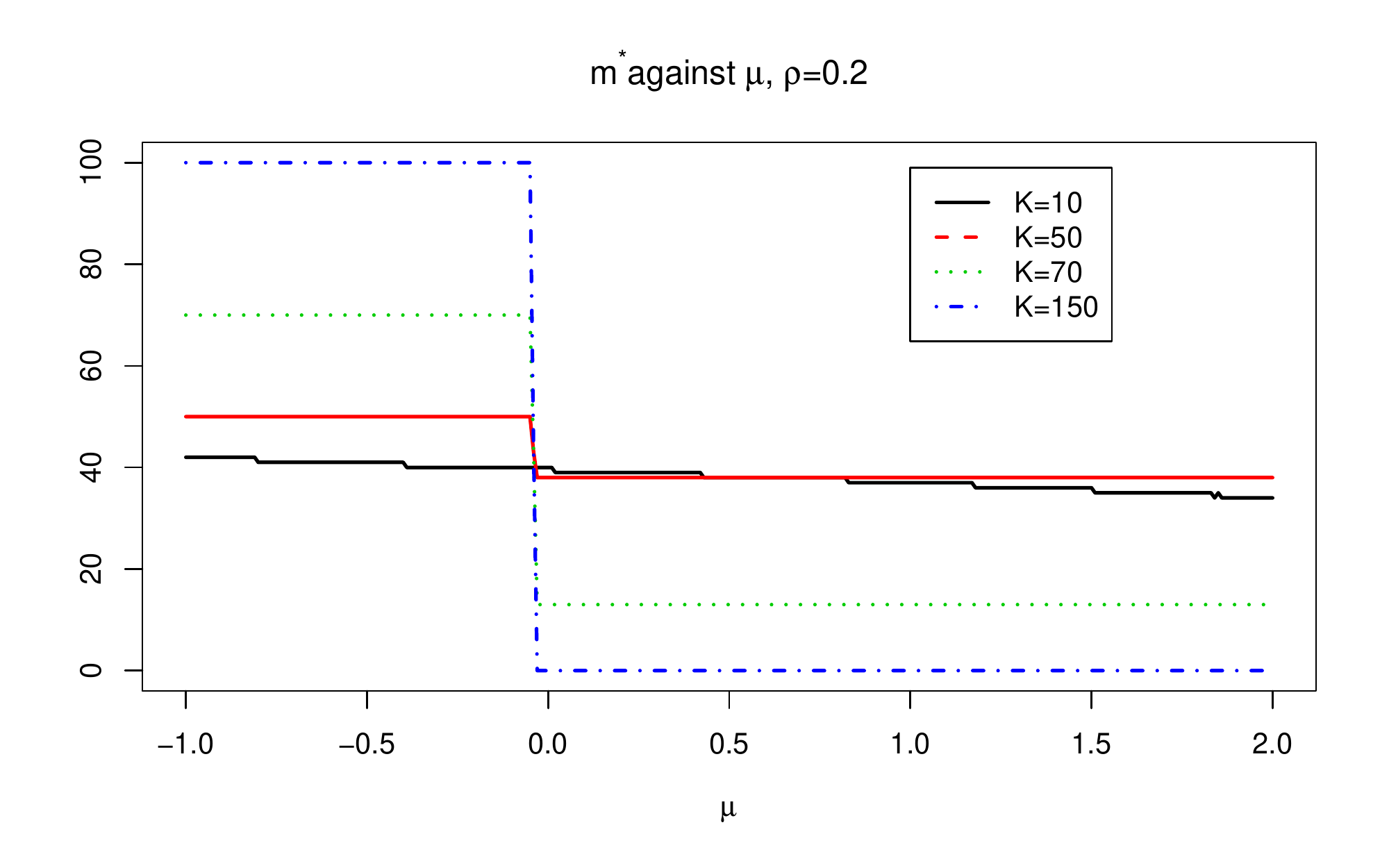}
	\includegraphics[width=0.495\textwidth]
	{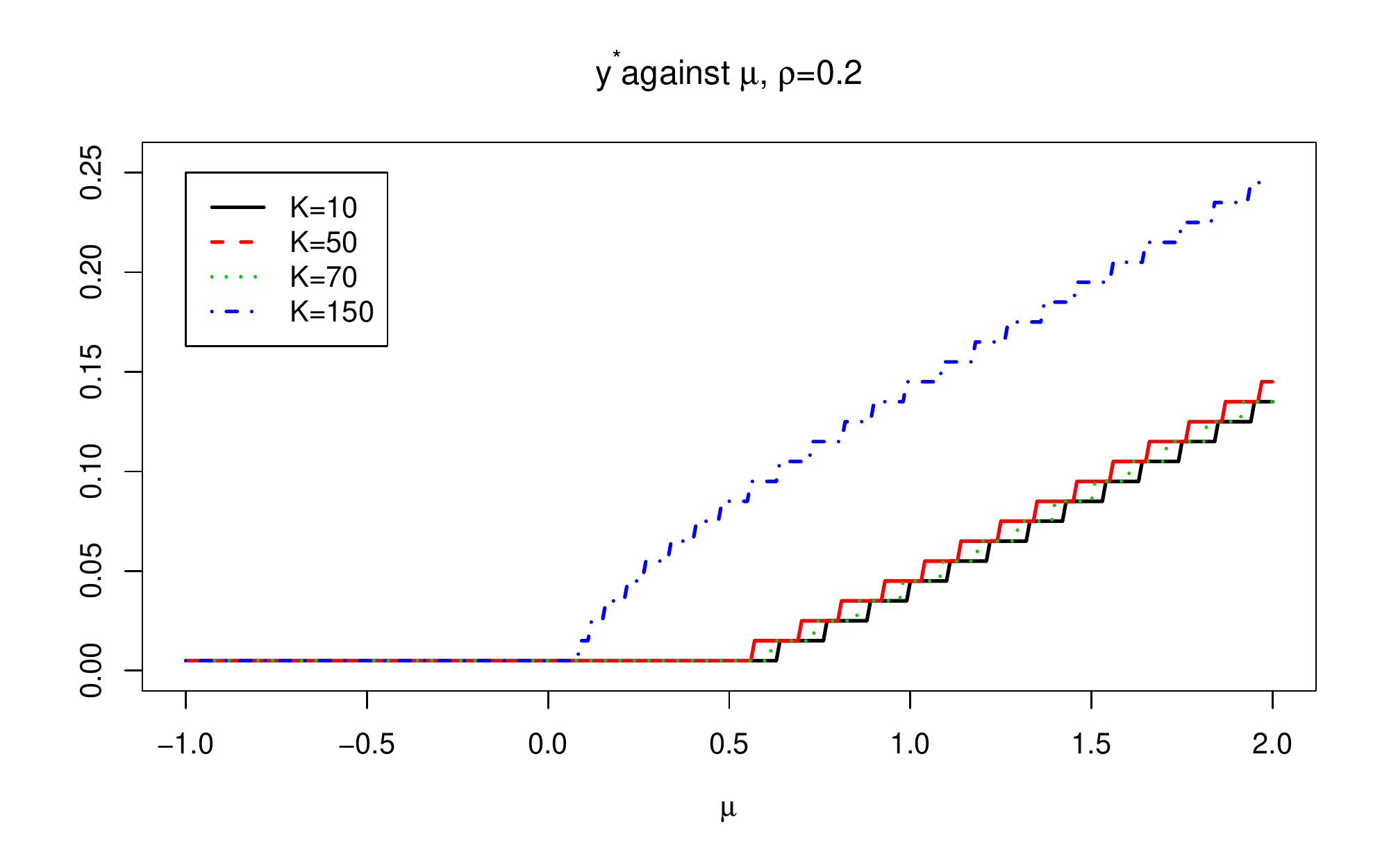}
	
	\caption{$m^*$ (Left) and $y^*$ (Right) against $\mu$ for $K=10$ (black solid line), $K=50$ (red dashed line), $K=70$ (green dotted line), $K=150$ (blue dot-dashed line). $\rho=0.2, \beta=0.01, M=100, r=0.003, f=0.003, \epsilon=0.01, T=0.1, S_0=100, d=\epsilon/2=0.005, \sigma=0.1$. }
	\label{mystar_againstMu_diffK}
\end{figure}

{
\par Theorem~\ref{thm:single:muNeg} shows that when $\mu<0$, $y^*=d$, which means that the investor may put a limit order at the best bid price. This is shown in behavior of $y^*$ in the right panel of Figure~\ref{mystar_againstMu_diffK}, $y^*=d$ when $\mu<0$. The left panel of Figure~\ref{mystar_againstMu_diffK} describes behavior of $m^*$, the optimal market order size at $t=0$. As shown in Theorem~\ref{thm:single:muNeg}, when $\mu<0$    $m^*=m^*_2$ when $K=10$, and $m^*=K$ for $K=50$, $70$, and $m^*=M$ when $K=150$. The behavior of $y^*$ and $m^*$ when $\mu>0$ is described in the following Theorem~\ref{thm:single:muPos}. In this theorem, we first introduce $T_0$, the threshold for the time horizon\footnote{Note that this work is about single-step. $T$ is not the number of the time steps, but it is just the length of time for the single step. } $T$ such that if the investor's time horizon  $T$  is bigger than $T_0$, then $y^*>d$. We also provide lower bound for $T_0$ and stepwise-linear approximation of $y^*$ as a function of $T$. 
}

{

\begin{thm}\label{thm:single:muPos}
	Let $(y^*, m^*)$ be defined as in Lemma~\ref{lem:bm:ECF}. Then, when $\mu>0$,

	$y^*>d$  for $T>T_0$. The lower bound of $T_0$ is that $$T_0>\frac{2d+f+r}{\mu} .$$ 

	Also, for $T>T_0$, as $T\searrow T_0 $, the first-order approximation for $y^*$ is given as:
	{
	$$ y^*(T)= d+ \kappa (T-T_0) + o ((T-T_0)^2),$$
	
	where

	$$\kappa:=\left(- \frac{\frac{\partial^2 ECF}{\partial T \partial y } (d, T_0 )}{\frac{\partial^2 ECF}{\partial y^2 } (d, T_0 )} \right). $$
	}
	The details of $\kappa$ are given in (\ref{kappa}).
	In addition, there exists $m^*$ as follows:
	%

\begin{enumerate}
	\item If $K\geq M,$  $m^*=0$.
	\item If $K< M$, $m^*$ is one of $\{0, M, m_4^*, m^*_5, m^*_6, M, |K-M|, M-K/(1-\rho) \}$
\end{enumerate}

where
$$
m^*_4=\frac{\epsilon_0}{(1+\epsilon_0)}M +\frac{(1-\epsilon_0)}{2(1+\epsilon_0)}K-\frac{\epsilon_4}{2(1+\epsilon_0)\beta},
$$

$$m^*_5=M-\frac{\beta(\epsilon_0+(1-\epsilon_0)(1-\rho))K+\epsilon_4}{2\beta(\epsilon_0+(1-\epsilon_0)(1-\rho)^2)}, 
\quad 
m^*_6=\frac{\beta (2M(\epsilon_0+(1-\rho)^2 (1-\epsilon_0)) +K(1-\epsilon_0)\rho)-\epsilon_4}{2\beta (1+\epsilon_0+(1-\rho)^2 (1-\epsilon_0))}.$$

$$\epsilon_4:=\left( \rho(1-\epsilon_0)(d+f+y+r)-2\rho (y-d)\tilde{N}(-\alpha_{T}) +\mu {T}(1-\rho+\rho\epsilon_0) \right), \quad
\epsilon_0= \left(N\left(\beta_{T} \right)-\tilde{N}\left(-\alpha_{T}\right) \right).$$
\end{thm}

}

\begin{proof}
\noindent\textbf{Part 1: Behavior of $y^*$: }First, note that as $y\to \infty$, from (\ref{ECF_rho_der_y.1}), (\ref{ECF_rho_der_y.2}), (\ref{ECF_rho_der_y.3}), $\partial ECF / \partial y\to 0^{-}$. Therefore, $y^*<\infty$.
	Next, let's find  $\left.\partial ECF / \partial y\right|_{y=d}$.

\begin{align}
\label{ECF_rho_der_y_d.1}
\left.\frac{\frac{\partial ECF}{\partial y}}{ (M-m)} \right|_{y=d}
&=
\rho\left(N\left(\frac{\mu {T}}{\sigma\sqrt{{T}}}\right)-{N}\left(-\frac{\mu {T}}{\sigma\sqrt{{T}}}\right)\right)\\
\label{ECF_rho_der_y_d.2}
&+\frac{2}{\sigma\sqrt{t}}
\left( \phi \left(-\frac{\mu {T}}{\sigma\sqrt{{T}}}\right) -\frac{\mu {T}}{\sigma\sqrt{{T}}} N\left(-\frac{\mu {T}}{\sigma\sqrt{t}}\right)\right)
\left(-\rho (2d+f+r-\mu {T})+ \beta \gamma \right) 
\end{align}
Note that (\ref{ECF_rho_der_y_d.1})$>0$, and (\ref{ECF_rho_der_y_d.2})$>0$ when 
\begin{equation}\label{bound:muT}
\mu T>2d+f+r+(-\beta \gamma)/\rho.
\end{equation}
Note that $\left.\partial ECF / \partial y\right|_{y=d}>0$  implies that $y^*>d$. Since from the proof of Theorem~\ref{thm:single:muNeg} we've shown that $\gamma<0$, $\mu T>2d+f+r+(-\beta \gamma)/\rho>2d+f+r$, so the lower bound for $T$ to satisfy this condition is $(2d+f+r)/\mu$. (we ignore $-\beta \gamma/\rho$ term since the definition of $\gamma$ contains $m$, so it can be correctly computed after we have the value of $m^*$. )

{

\begin{align}\label{ECF_rho_der_y_d.3}
	\frac{\partial}{\partial t} \left(\left.\frac{\frac{\partial ECF}{\partial y}}{ (M-m)} \right|_{y=d}\right)
	&=
	-\gamma \beta \left(  \frac{1}{\sigma {T} \sqrt{{T}}}(\phi(h_t) - h_t N(-h_{T}) )+\frac{\mu }{\sigma^2 t}N(-h_{T}) \right) \\
	\label{ECF_rho_der_y_d.4}
	+&\frac{\rho}{\sigma {T} \sqrt{{T}}}\bigg( {(2d+f+r) }
		\left( (\phi(h_{T}) - h_t N(-h_{T}) )+ N(-h_{T}) \frac{\mu T}{\sigma \sqrt{T}}\right)+{2\mu T}{} (\phi(h_{T}) - h_{T} N(-h_{T}) )\bigg)
\end{align}

Note that both (\ref{ECF_rho_der_y_d.3}), (\ref{ECF_rho_der_y_d.4}) are positive. This implies that $\left.\frac{\partial ECF}{\partial y}\right|_{y=d}$ is a increasing function of $T$, which supports that for $T_0$ such that $\left.\frac{\partial ECF}{\partial y} \right|_{y=d,T=T_0}=0 $, $\left.\frac{\partial ECF}{\partial y}\right|_{y=d}>0$ for any $T>T_0$. 


Now, we will use the mean value theorem to show the behavior of $y^*$ when $T$ is close to $T_0$. To this end, the following conditions are necessary: $\left|\partial_y^2 ECF\right|$ needs to be strictly positive at $y=d,T=T_0$. 
$\partial_y^2 ECF (y=d,T=T_0)=-2\rho (M-m) \mu N(-h_t)  /\sigma^2 $, so $|\partial_y^2 ECF (y=d,T=T_0)|>0$. Next, since $y^*$ satisfies $\partial_y ECF(y^*)=0$, and thus, by the Implicit Function Theorem, there exists an open set $U$ containing $y=d$, an open set $V$ containing $T=T_0$, and a unique continuously differentiable function $y^*(T)$ such that
$$
\{ (y^*(t),T | T\in V \} = \left\{  (y,T)\in U\times V \left| \frac{\partial ECF}{\partial y} (y,T)=0 \right. \right\}.
$$
In particular, $y^*(T)\to d$ as $T\to T_0$. Furthermore, since $T_0>0$, it is clear that $\partial_y ECF$ is differentiable in a neighborhood of $(y=d,T=T_0)$, and, thus, we can apply the mean value theorem to show that there exists $\delta \in (0,1)$ such that

$$
0=\frac{\partial ECF}{\partial y } (y^*(T),T) = \frac{\partial^2 ECF}{\partial y^2 } (d+\delta (y^* (T)-d), T_0 + \delta (T-T_0) ) y^* (T) + \frac{\partial^2 ECF}{\partial T \partial y } (d+\delta (y^* (T)-d), T_0 + \delta (T-T_0) ) (T-T_0).
$$
Since $\partial^2 ECF / \partial y^2$, $\partial^2 ECF / \partial y \partial t $ are both continuous when $T>0$, and there is an open set containing $(d,T_0)$ such that $\partial^2 ECF / \partial y^2$ is strictly negative and, furthermore,

$$
\frac{y^*(T)-d}{T-T_0}=- \frac{\frac{\partial^2 ECF}{\partial T \partial y } (d+\delta (y^* (T)-d), T_0 + \delta (T-T_0) )}{\frac{\partial^2 ECF}{\partial y^2 } (d+\delta (y^* (T)-d), T_0 + \delta (T-T_0) )} \quad \overrightarrow{T\to T_0} \quad \left(- \frac{\frac{\partial^2 ECF}{\partial T \partial y } (d, T_0 )}{\frac{\partial^2 ECF}{\partial y^2 } (d, T_0 )} \right):=\kappa(T_0),
$$
%
%


\begin{equation}\label{kappa}
\kappa({T_0})= \frac{{3 \phi(-h_{T_0})}\sigma  + \frac{ \sigma^2(N(h_{T_0}) - N(-h_{T_0})) }{2\mu\sqrt{T_0}}+{2 (-h_{T_0} N(-h_{T_0}))} \left(
	-\frac{2d+f+r}{\sqrt{T_0}}+\mu \sqrt{T_0}+\frac{\beta\gamma}{\rho\sqrt{T_0}}	
	+ \sigma \right)}{{2 N(h_{T_0})\sqrt{T_0} }
}.
\end{equation}

}
\begin{remk}
	
	Note that  $\kappa$ contains $\gamma$, which depends on the value of $m$. For the fast approximation, we may use the lower bound of $\kappa$, 
	
	\begin{equation}\label{kappa.lb}\nonumber
	\underline{\kappa}({T_0})= \frac{{3 \phi(-h_{T_0})}\sigma  + \frac{ \sigma^2(N(h_{T_0}) - N(-h_{T_0})) }{2\mu\sqrt{T_0}}+{2 (-h_{T_0} N(-h_{T_0}))} \left(
		-\frac{2d+f+r}{\sqrt{T_0}}+\mu \sqrt{T_0}	
		+ \sigma \right)}{{2 N(h_{T_0})\sqrt{T_0} }
	}.
	\end{equation}
	
\end{remk}

\noindent\textbf{Part 2: Behavior of $m^*$}\\
\par Now, to find $m^*$, let us reorganize $ECF(y,m)$ as following:

\begin{gather}
	\begin{aligned}
ECF(y,m)=
M \epsilon_5-m \epsilon_4 -\beta\bigg[& (M-m)\big(\epsilon_0  (K-M+m)^{-}+(1-\epsilon_0)(1-\rho)(K-(M-m)(1-\rho))^{-}\big)\\
&+m (K-m)^{-} \bigg],
\end{aligned}\label{thm2:mupos:ecf}
\end{gather}

%
%

where

$$\epsilon_4:=\left( \rho(1-\epsilon_0)(d+f+y+r)-2\rho (y-d)\tilde{N}(-\alpha_{T}) +\mu {T}(1-\rho+\rho\epsilon_0) \right), $$ 

$$\epsilon_5:=\left(S_0 -2\rho(y-d)\tilde{N}(-\alpha_{T})-(d+f-\mu {T})(1-\rho+\rho\epsilon_0)+(1-\epsilon_0)\rho(y+r)\right).$$

Note that $\epsilon_5$ does not depend on m, and $\epsilon_4>0$ since it can be reorganized as

$\left( \rho(1-\epsilon_0)(d+f+r)+2\rho (d)\tilde{N}(-\alpha_{T}) +\mu {T}(1-\rho+\rho\epsilon_0) \right) +\rho y (N(-\beta_{T})-\tilde{N}(-\alpha_{T}))$, and $N(-\beta_{T})-\tilde{N}(-\alpha_{T})>0$. 
Now, investigate the behavior of (\ref{thm2:mupos:ecf}) in different range of $m$ to find $m^*$.  

\noindent\textbf{Case 1-1: $(K-m)\geq 0, (K-M+m)\geq 0, (K-(M-m)(1-\rho))\geq 0$.}
In this case, the (\ref{thm2:mupos:ecf}) becomes a linear decreasing function of $m$. When $K\geq M$, for any $m$ this condition is satisfied so $m^*=0$.

\noindent\textbf{Case 1-2: $(K-m)< 0, (K-M+m)\geq 0, (K-(M-m)(1-\rho))\geq 0$}
In this case, (\ref{thm2:mupos:ecf}) attains its maximum at $\frac{K}{2}-\frac{\epsilon_4}{2\beta}$, which is $<K$, so in this range (\ref{thm2:mupos:ecf}) is a decreasing function of $m$. 

\noindent\textbf{Case 2-1: $(K-m)\geq  0, (K-M+m)< 0, (K-(M-m)(1-\rho))\geq 0$}
In this case, (\ref{thm2:mupos:ecf}) attains its maximum at 
$m^*_4=M-K/2-\epsilon_4/2\beta\epsilon_0$,


\noindent\textbf{Case 2-2: $(K-m)<  0, (K-M+m)< 0, (K-(M-m)(1-\rho))\geq 0$}
In this case, (\ref{thm2:mupos:ecf}) attains its maximum at

$m^*_5=\frac{\epsilon_0}{(1+\epsilon_0)}M +\frac{(1-\epsilon_0)}{2(1+\epsilon_0)}K-\frac{\epsilon_4}{2(1+\epsilon_0)\beta}
$


\noindent\textbf{Case 3: $(K-M+m)\geq 0, (K-(M-m)(1-\rho))< 0$}
This condition is contradictory. \\
\noindent\textbf{Case 4-1: $(K-m)\geq  0, (K-M+m)< 0, (K-(M-m)(1-\rho))< 0$}
In this case, (\ref{thm2:mupos:ecf}) attains its maximum at
$m^*_5=M-\frac{\beta(\epsilon_0+(1-\epsilon_0)(1-\rho))K+\epsilon_4}{2\beta(\epsilon_0+(1-\epsilon_0)(1-\rho)^2)}$


\noindent\textbf{Case 4-2: $(K-m)<  0, (K-M+m)<0, (K-(M-m)(1-\rho))< 0$}
%
In this case, (\ref{thm2:mupos:ecf}) attains its maximum at
$m^*_6=\frac{\beta (2M(\epsilon_0+(1-\rho)^2 (1-\epsilon_0)) +K(1-\epsilon_0)\rho)-\epsilon_4}{2\beta (1+\epsilon_0+(1-\rho)^2 (1-\epsilon_0))}$

To summarize, when $K\geq M$, $m^*=0$ . If not, the maximum will be obtained at one of $\{0, M, m_4^*, m^*_5, m^*_6, M, M-K, |M-K/(1-\rho)| \}$.

\end{proof}

{

In Theorem~\ref{thm:single:muPos}, we have shown the behavior of $y^*$ and $m^*$ when $\mu$ is positive. 

\begin{figure}
	\includegraphics[width=0.495\textwidth]{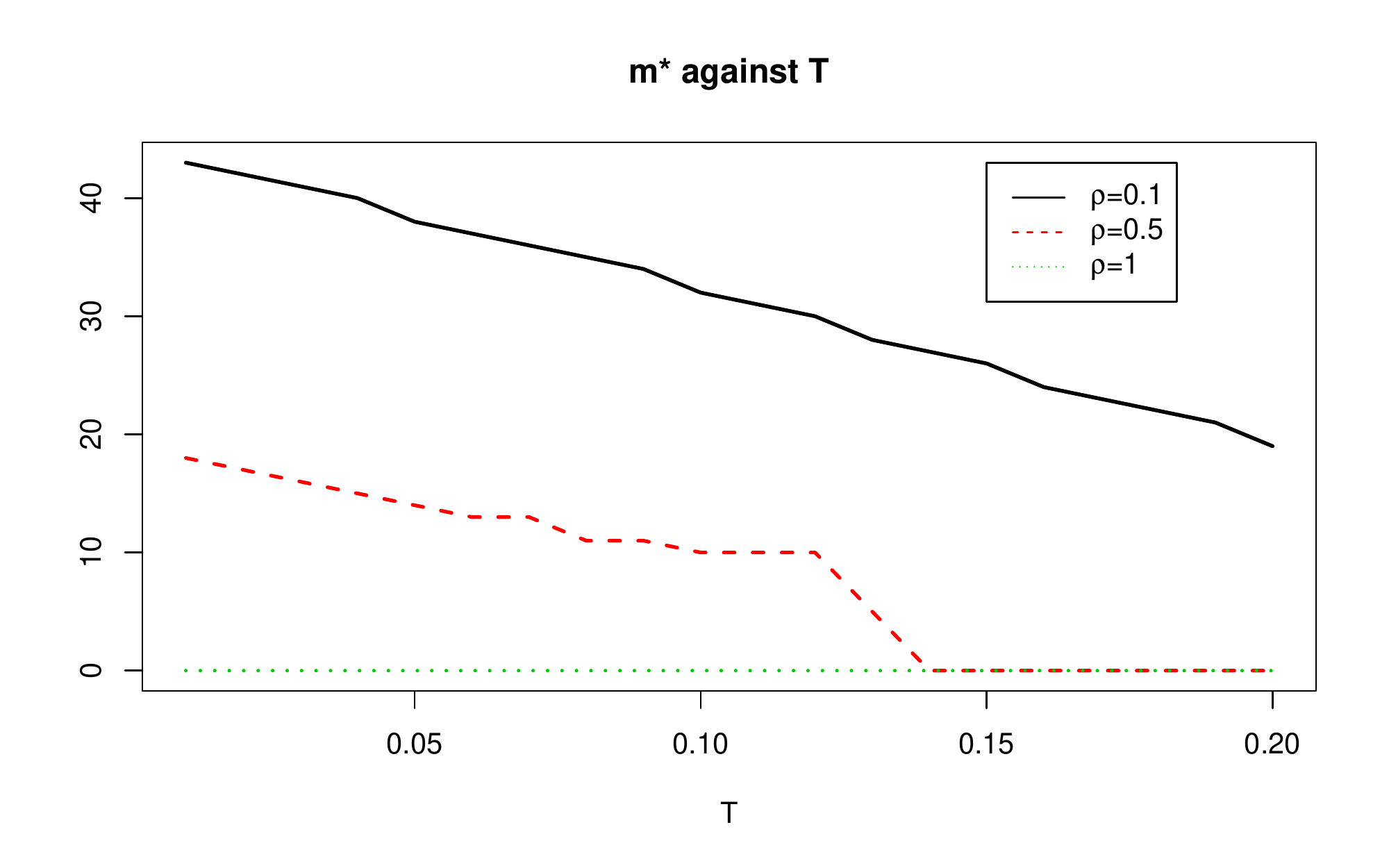}
	\includegraphics[width=0.495\textwidth]{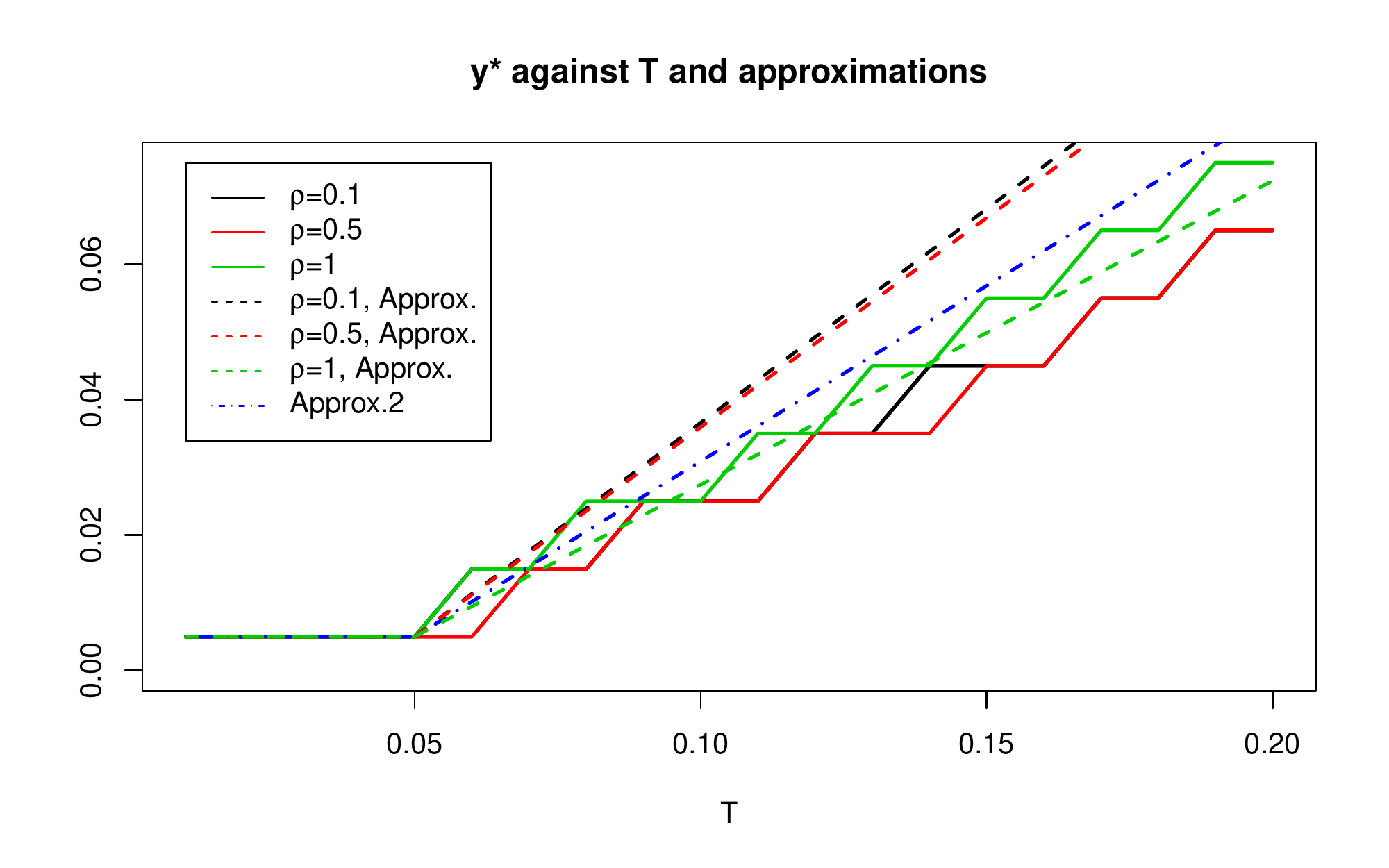}
	
	\caption{Behavior of $m^*$ (Left) and $y^*$ (Right) against T. For the left graph, black solid line is when $\rho=0.1$, red dashed line is when $\rho=0.5$, green dotted dot-dashed line is when $\rho=1$. For the graph in the right panel, $y^*$ and it's first-order approximation $d+\kappa(T-T_0)$ are described in solid and dashed lines, respectively, using Black color($\rho=0.1$), red color($\rho=0.5)$, and green color ($\rho=1$). Blue dot-dashed line in Right panel is using the approximation from \cite{FLP:2018}. For both panels, $M=100, r=0.003, f=0.003, \epsilon=0.01, S_0=100, d=\epsilon/2=0.005, \sigma=0.1, \mu=0.5$ .  }
	\label{mystar_againstT_diffRho}
\end{figure}

\begin{figure}
	\includegraphics[width=0.495\textwidth]{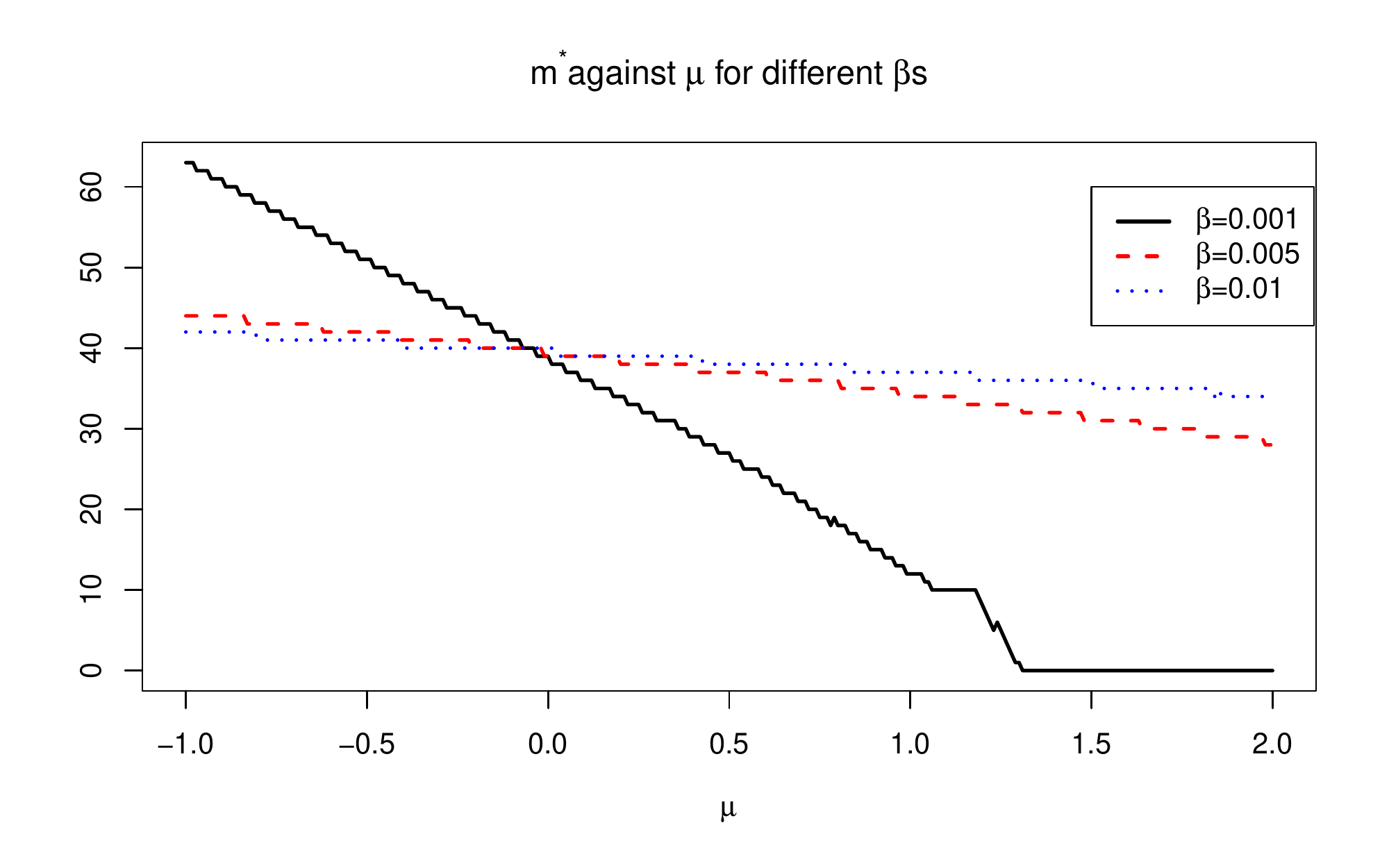}
	\includegraphics[width=0.495\textwidth]{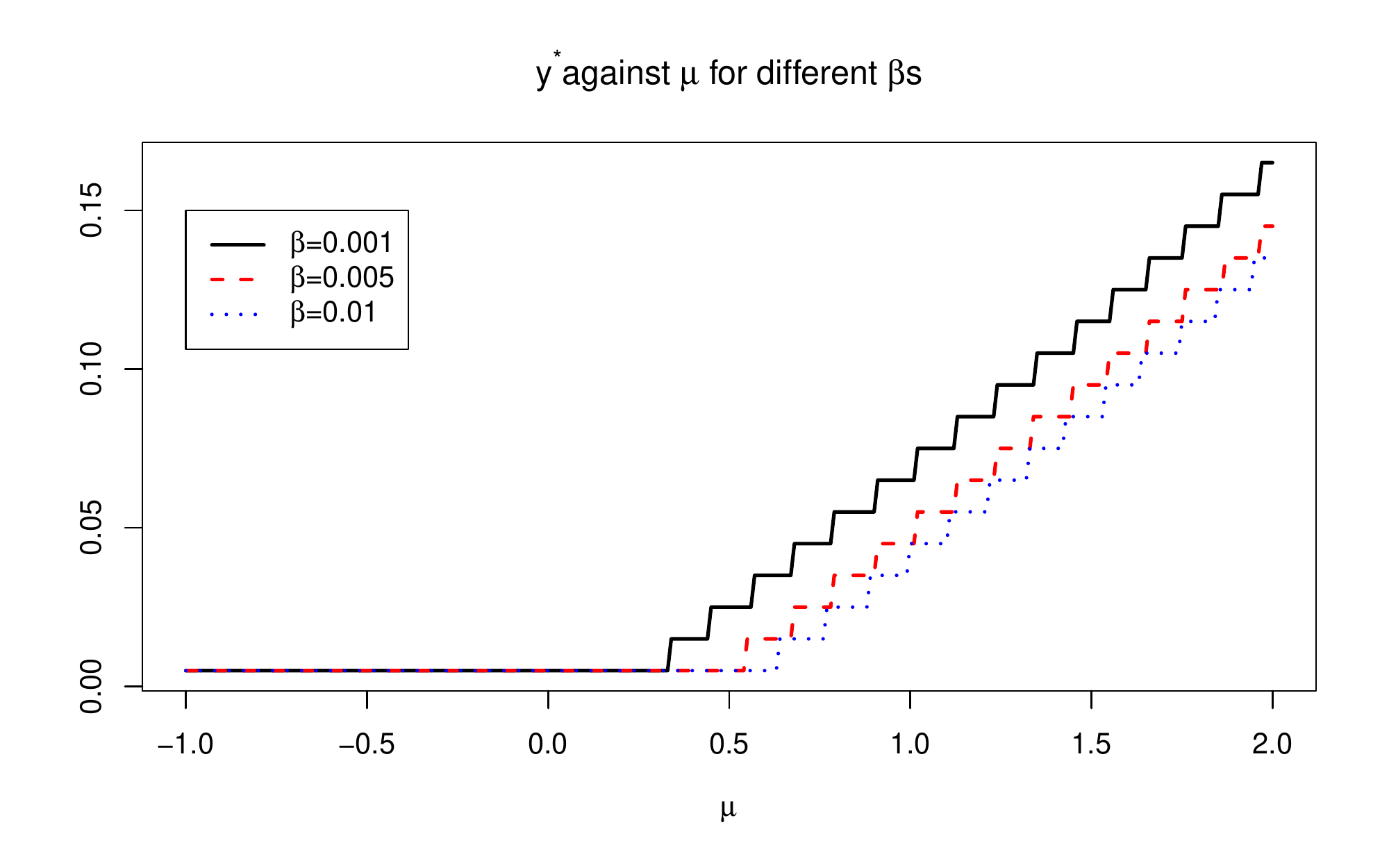}
	\caption{
		$m^*$  (left) and $y^*$ (right) against $\mu$ for $\beta=0.001$ (black solid line), $\beta=0.005$ (red dashed line), and $\beta=0.01$ (blue dotted line). $K=10, \rho=0.2, $  $M=100, r=0.003, f=0.003, \epsilon=0.01, d=\epsilon/2=0.005, T=0.1, S_0=100, \sigma=0.1$.}
	\label{mystar_againstMu_diffBetas}
\end{figure}

\begin{figure}
	\includegraphics[width=0.495\textwidth]{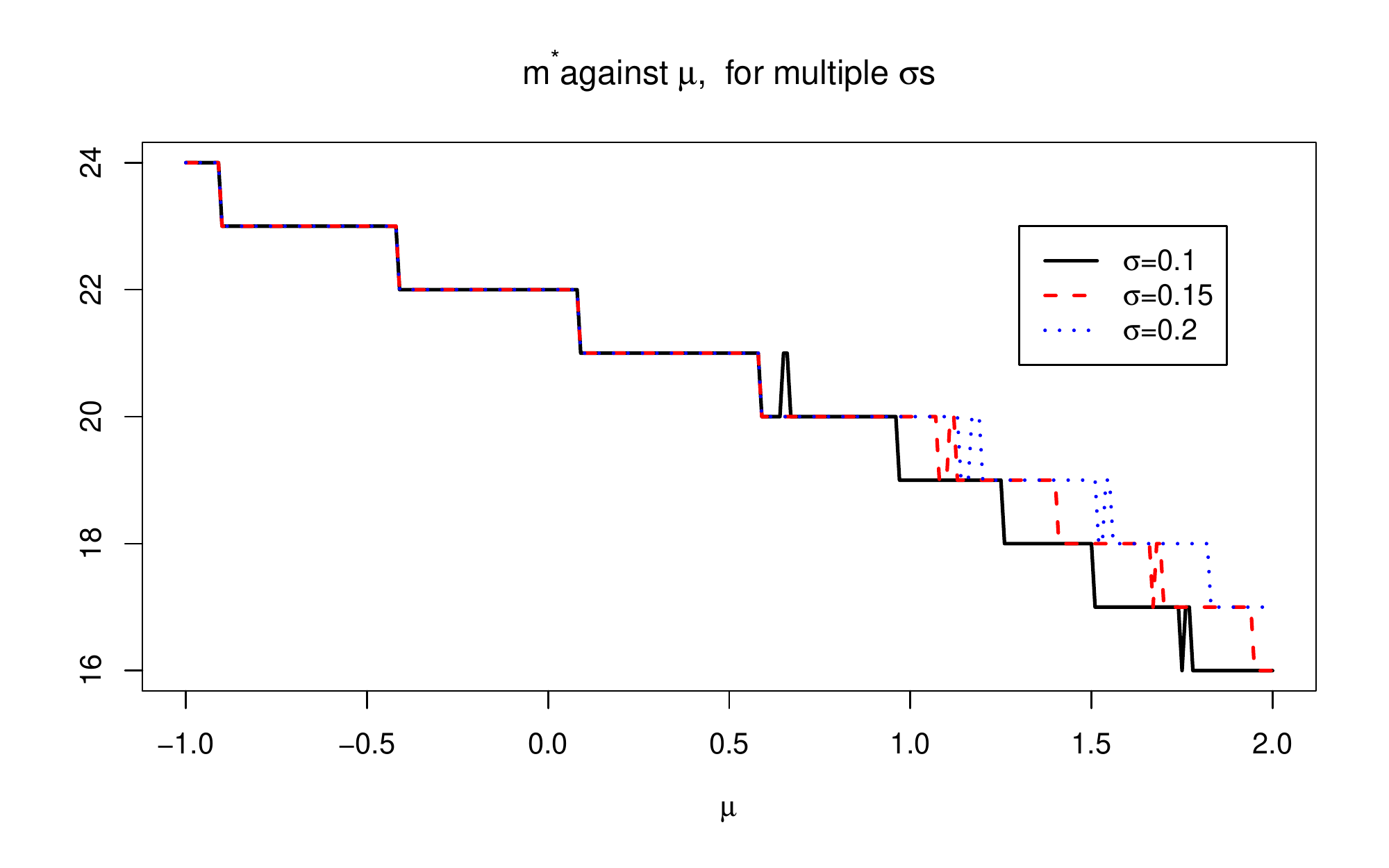}
	\includegraphics[width=0.495\textwidth]{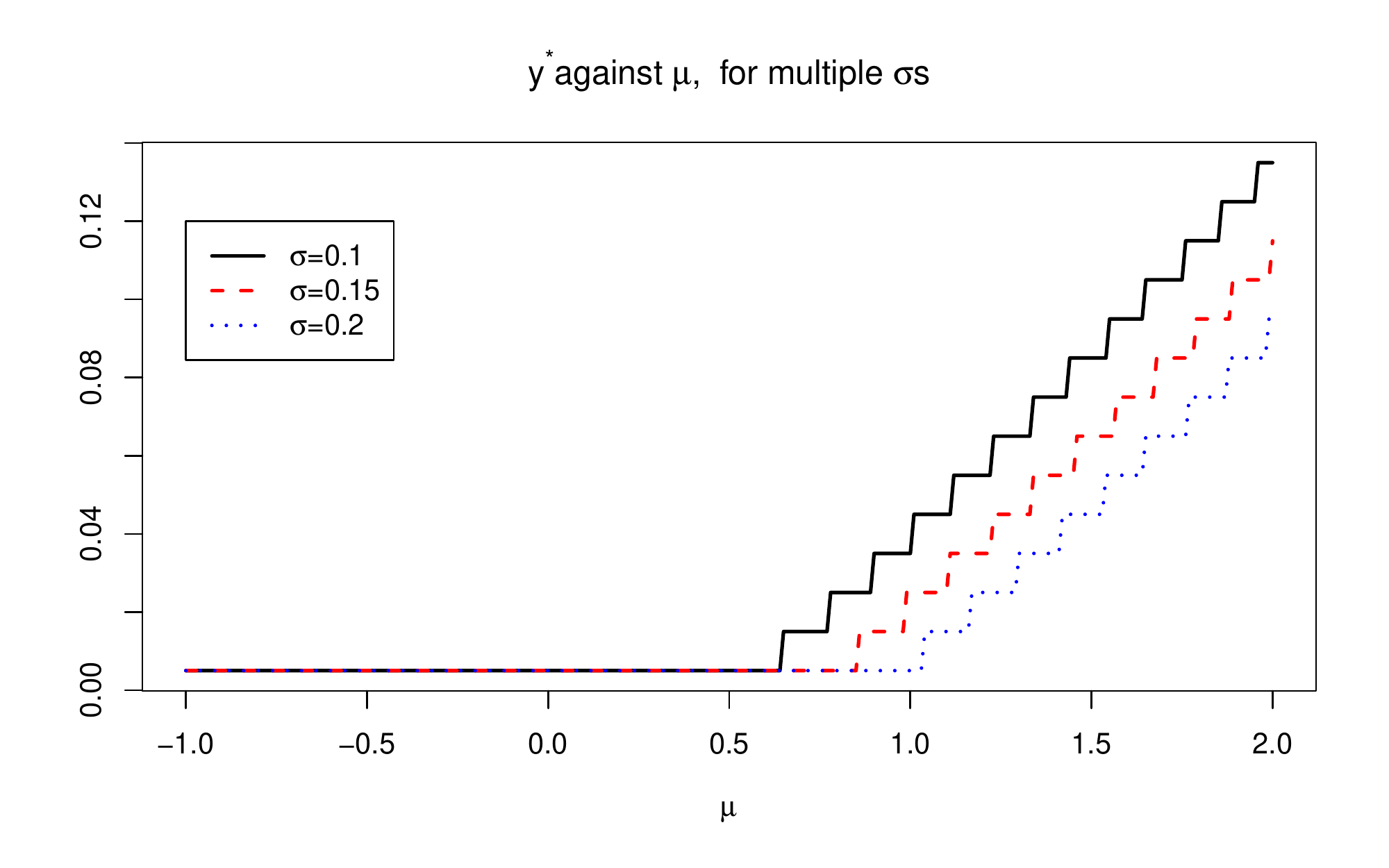}
	\caption{ $m^*$ (left), and $y^*$ (right) against $\mu$ for $\sigma=0.1$ (black solid line), $\sigma=0.15$ (red dashed line), $\sigma=0.2$ (blue dotted line).  $K=10$, $M=100, r=0.003, f=0.003, \epsilon=0.01, d=\epsilon/2=0.005, T=0.1, S_0=100, \rho=0.5$. 
	}
	\label{mystar_againstMu_diffSig_BM}
\end{figure}

\begin{figure}
	
	\includegraphics[width=0.495\textwidth]{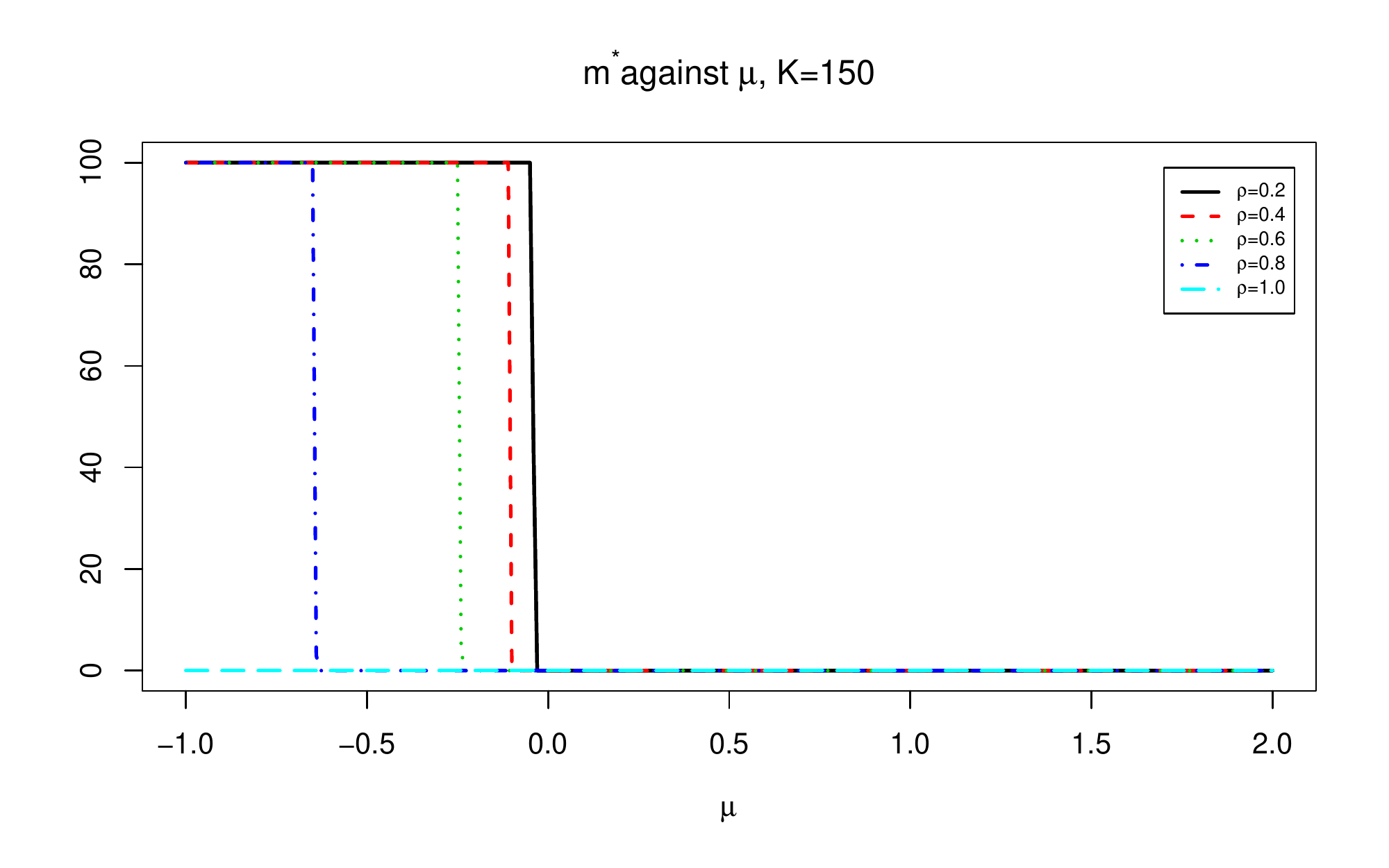}
	\includegraphics[width=0.495\textwidth]{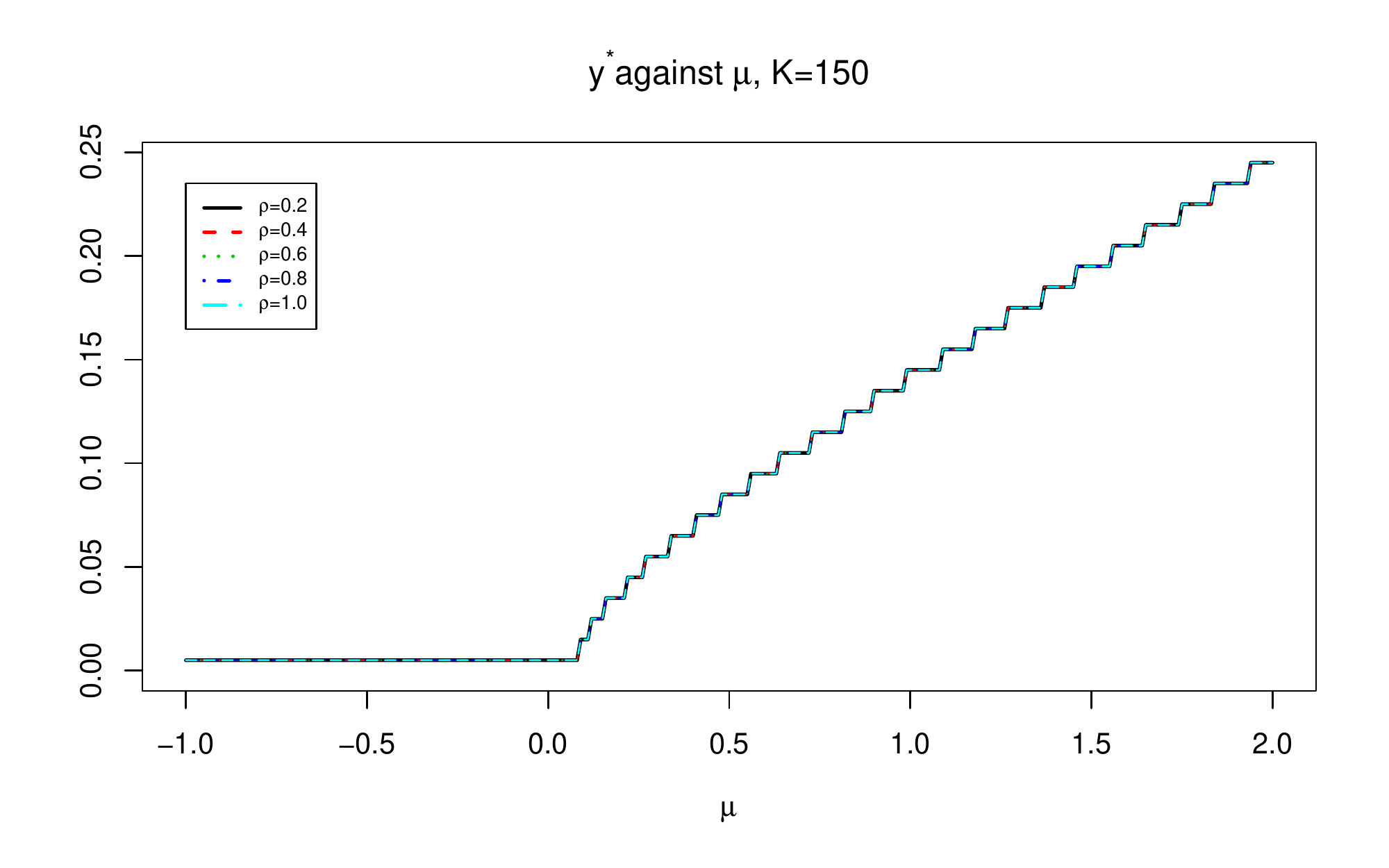}	
	\includegraphics[width=0.495\textwidth]{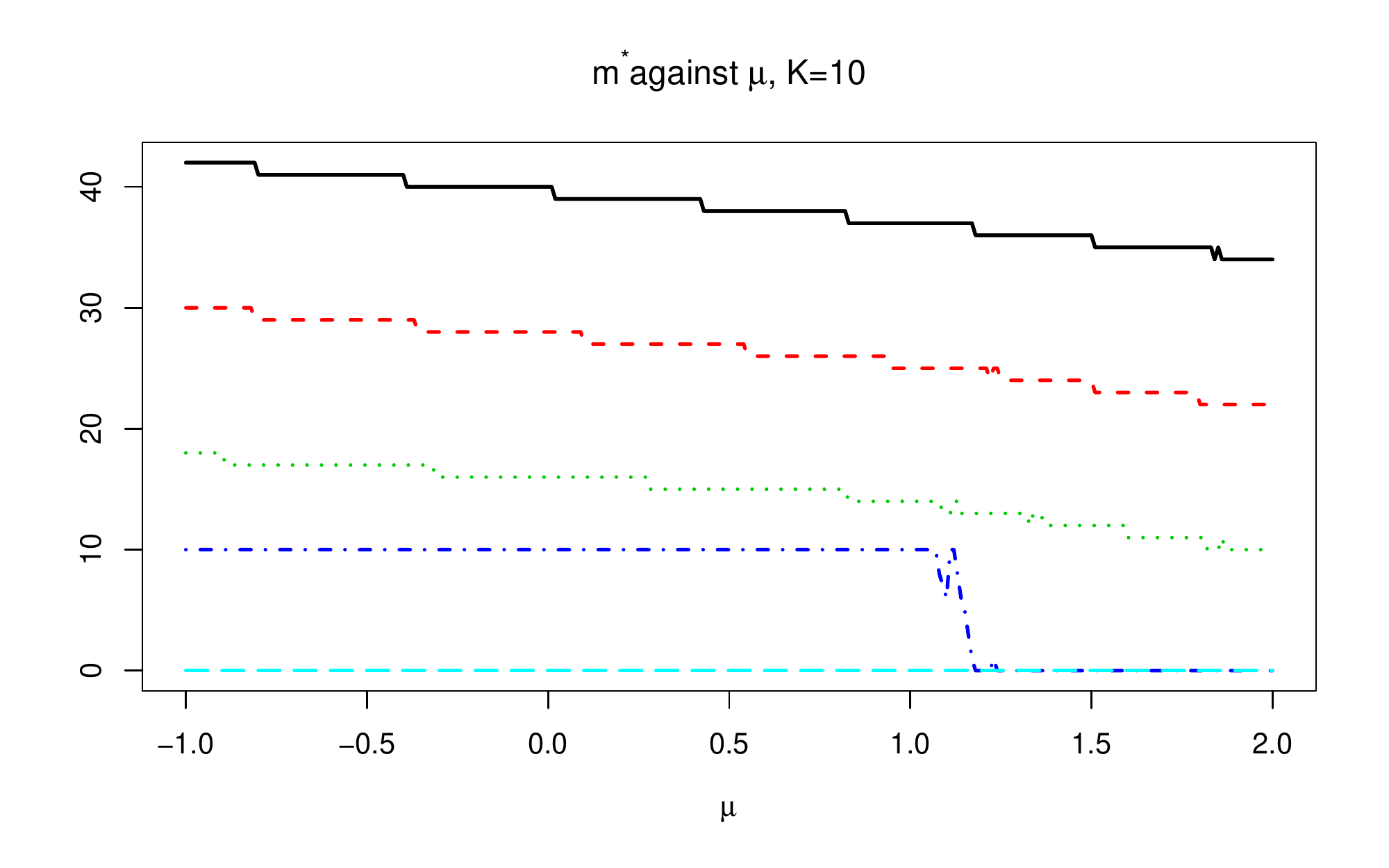}
	\includegraphics[width=0.495\textwidth]{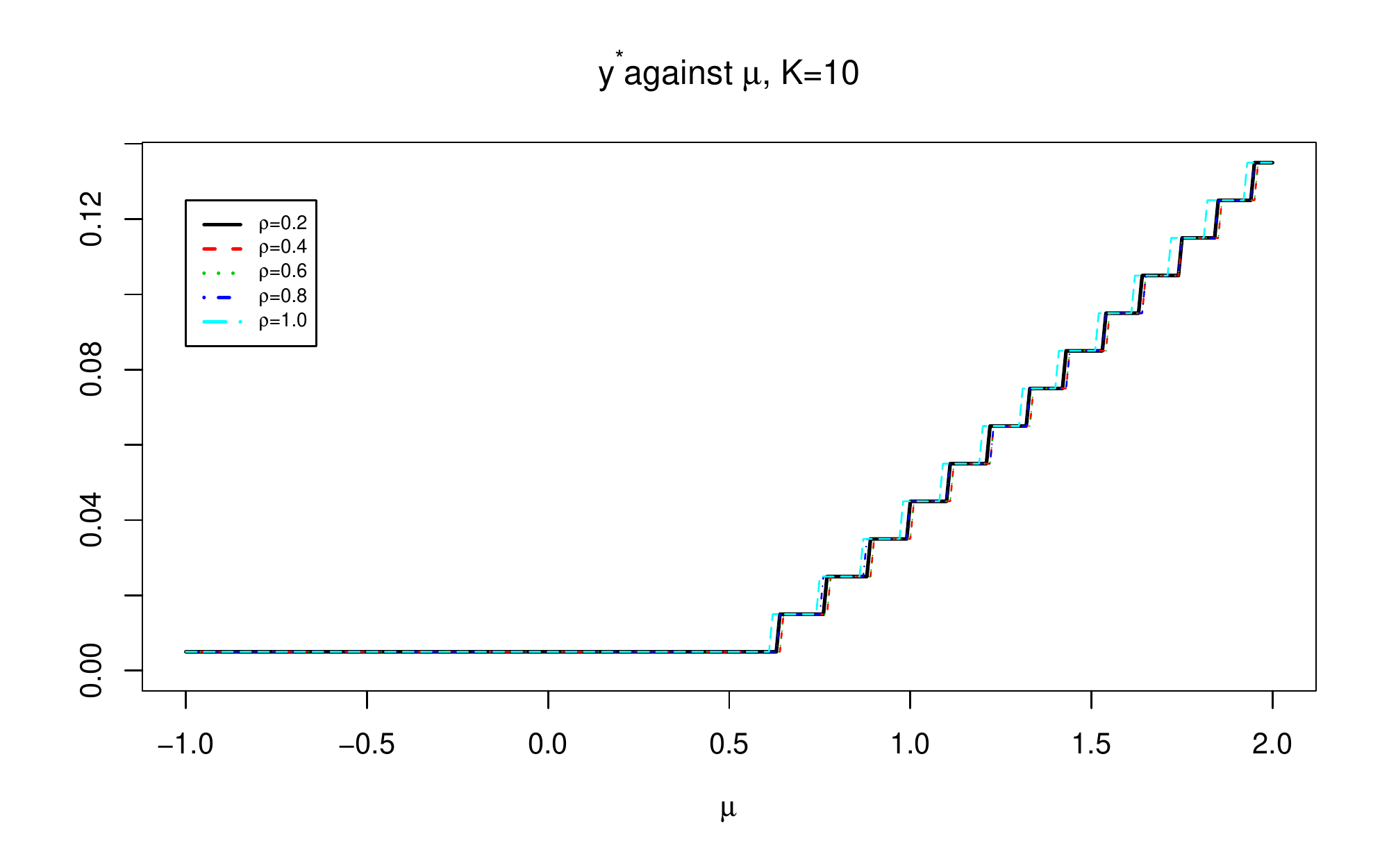}
	\includegraphics[width=0.495\textwidth]{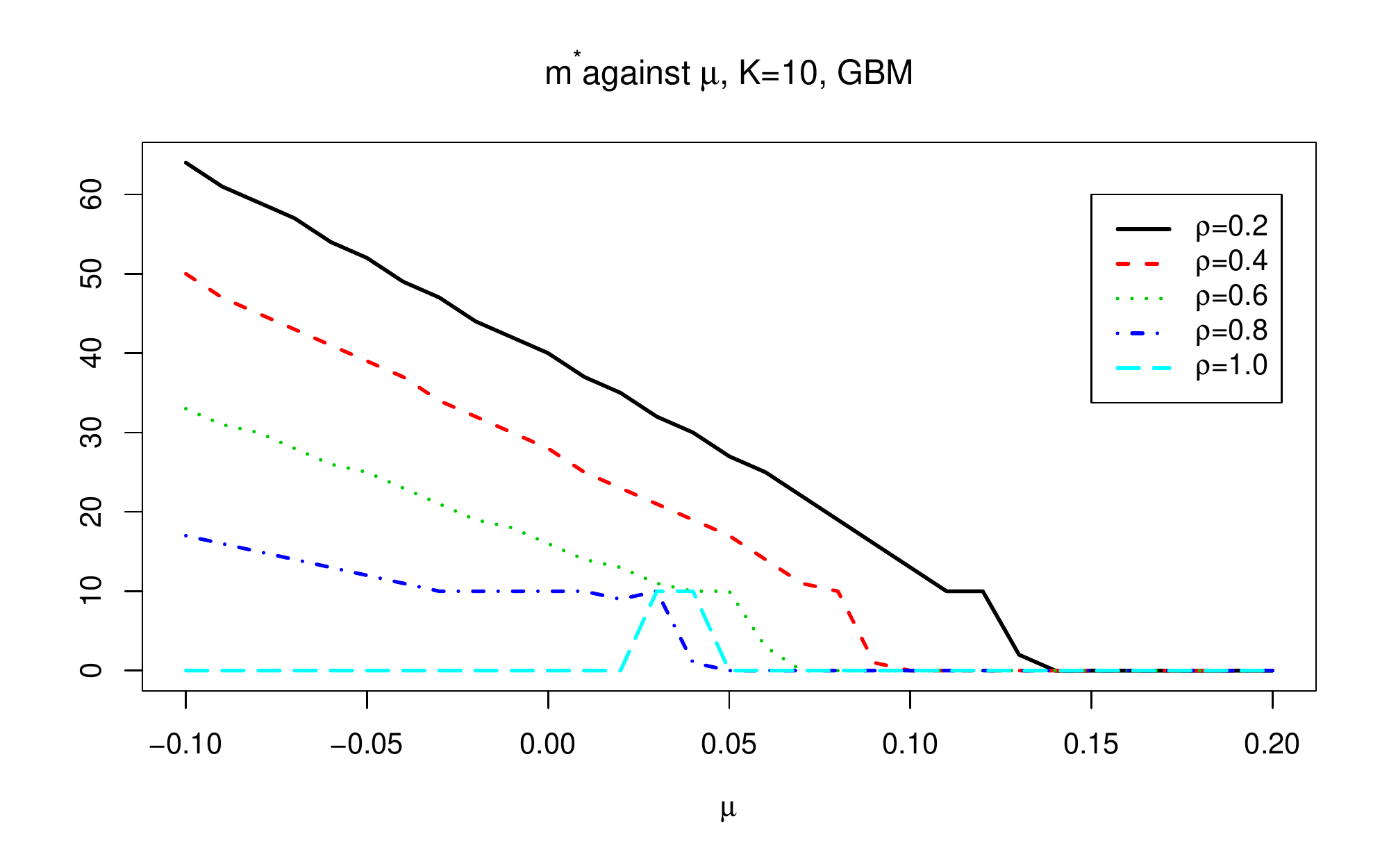}
	\includegraphics[width=0.495\textwidth]{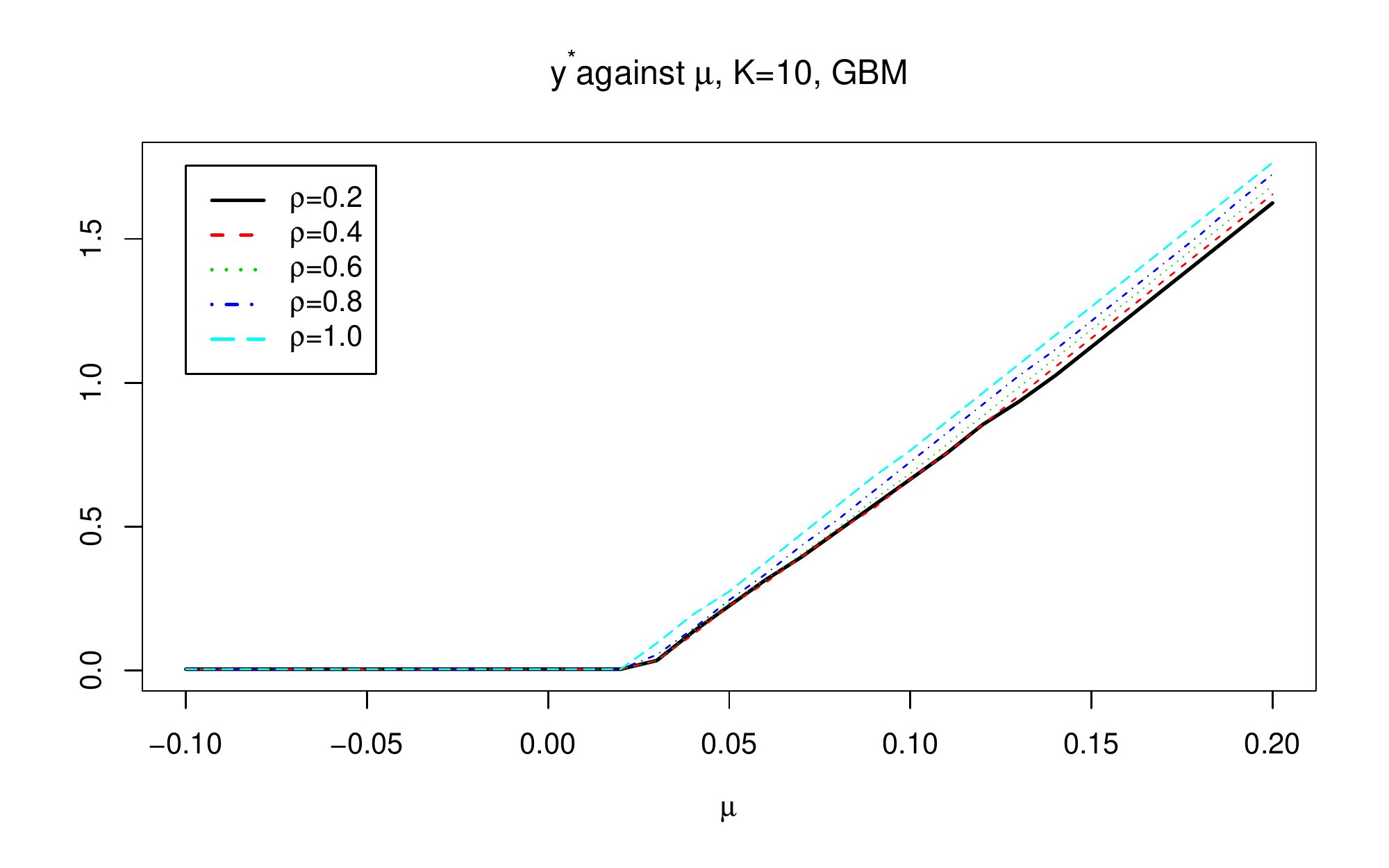}
	
	\caption{ $m^*$ (left), $y^*$ (right) against for $K=150$, BM (top), $K_10$, BM (middle) and $K=10$, GBM (bottom). $\rho=0.2$ (black solid line), $\rho=0.4$ (red dashed line), $\rho=0.6$ (green dotted line), $\rho=0.8$ (blue dot-dashed line) and $\rho=1$ (cyan long dahsed line). $M=100, r=0.003, f=0.003, \epsilon=0.01, d=\epsilon/2=0.005, T=0.1, S_0=100, \sigma=0.1$ for BM (top, middle) and $\sigma=0.01$ for GBM(bottom).}
	\label{mystar_againstMu_diffRhos}
\end{figure}

\par First, let us focus on the behavior of $y^*$. As shown in the right panel of Figure~\ref{mystar_againstMu_diffK}, $y^*=d$ until $\mu$ reaches a certain positive value (the lower bound is $\mu T>2d+f+r+(-\beta \gamma)/\rho$)  and then linearly increases with respect to $\mu$. 

\par Also, in the right panel of Figure~\ref{mystar_againstT_diffRho} it is observed that $y^*=d$ until a certain time threshold,$T_0$, and then increases linearly with respect to $T$. 

\par This complies with the result of Theorem~\ref{thm:single:muPos} that $y^*>d$ for a big enough $\mu T$. In Theorem~\ref{thm:single:muPos}, we have provided the existence of such a threshold, and named the threshold in the time horizon as $T_0$. We provided the stepwise-linear approximation of $y^*$ as a function of $T$: $y^*(T)\approx d+\kappa (T-T_0)$, the first-order approximation given in Theorem~\ref{thm:single:muPos}. These approximations are also found in the right panel of Figure~\ref{mystar_againstT_diffRho}. Each $y^*$ is described in solid lines, and the approximation is given in the dashed line using matching colors. Comparing with the result from \cite{FLP:2018}, the blue dotted line in the Right panel of Figure~\ref{mystar_againstT_diffRho} is showing first order Taylor approximation of $y^*$ using Theorem 3.6. of \cite{FLP:2018}. These approximations of $y^*$ converges close to the actual $y^*$ as $T\searrow T_0$.  } These approximation methods provide the quick computation method that the investor can use to decide the limit order placement based on their time horizon, $T$.

\par Now, let us focus on the behavior of $y^*$ with various $K$s, the initial market depth. The right panel of Figure~\ref{mystar_againstMu_diffK} shows that $y^*$ increases as $K$ increases, but the difference is relatively small for $K=10, 50, 70$ and there is huge difference when $K=150$. The case of $K=150$ is special since this is only the case when $K>M$ ($M=100$ in the simulation), which represents the case when the market depth $K$ is bigger than the investor's inventory $M$. In this case, the investor would not fear to put limit orders at the higher level, since even if limit orders do not get  executed until time $T$, the market order at time $T$ still will be beneficial since there is no impact on the supply curve based on the quantity of the market order size $M$.

\par Another interesting analysis of $y^*$ is the change of behavior for different $\beta$s. Figure~\ref{mystar_againstMu_diffBetas} shows behavior of $y^*$ int he right panel. $y^*$ increases as $\beta$ decreases. Also, recall that there is a lower bound of $\mu$ such that $y^*>d$ after a certain value of $\mu$, where the bound is $\mu T>2d+r+f+(-\beta \gamma)/\rho$ (given in  Equation (\ref{bound:muT})), where $\gamma<0$. Therefore, as $\beta$ increases, the bound also increases, which is observed in Figure ~\ref{mystar_againstMu_diffBetas}. 

\par In addition, the behavior of $y^*$ for various values of $\sigma$, the volatility of price process is shown in Figure~\ref{mystar_againstMu_diffSig_BM}.  As $\sigma$ increases, after $\mu$ gets bigger than a certain (positive) threshold value, $y^*$ decreases. While for low (positive) $\mu$, bigger $\sigma$ implies a higher chance of the limit order execution. However it is not necessarily true for high $\mu$, and this explains why $y^*$ decreases as $\sigma$ increases for big enough $\mu$.

\medskip
\par  Now, let us show the behavior of $m^*$. Figure~\ref{mystar_againstT_diffRho} shows the behavior of $m^*$ against $T$ when $\mu>0$. The behavior of $m^*$ follows the result of Theorem~\ref{thm:single:muPos}, and it is decreasing as $T$ increases. This is because when $\mu>0$ and $T$ increases, the benefit from the limit order increases. When $\mu>0$, the probability of limit order execution increases as $T$ increases, and even if a limit order is not executed, the expected value of $S_T$ is higher than $S_0$, so the market order at $t=0$ is less preferred.

\par The left panel of Figure~\ref{mystar_againstMu_diffRhos} shows the behavior of $m^*$ against $\mu$ when $K=150$ (top) and $K=10$ (middle) for different $\rho$'s. $K=10$ represents the case when the market depth is low (the quantity of limit order on best bid queue when $t=0$ is low). In this case, $m^*$ shows a relatively smooth decreasing trend as $\mu$ increases.
 For $K=150$, on the other hand, $m^*$ is still decreasing but there is a sudden drop from $m^*=M$ to $m^*=0$. When the market depth is high, then all the market order is covered by the current existing limit order. In this case, our inventory's size does not affect the supply curve. In other words, when $K>M$, the conclusion for $m^*$ and $y^*$ should be the same whether M=1 or M=K-1. So in this case, the investor's choice comes down to two choices: $m=0$ (just use limit order) and $m=M$ (just use market order).

\par The behavior of $m^*$ for different $K$s with a fixed $\rho$ is shown in the right panel of Figure~\ref{mystar_againstMu_diffK}. $m^*$,  the optimal quantity for the market order, decreases as $\mu$ increases. It is because as $\mu$ increases, there is higher chance of $LO$ execution and the market order at $T$ is more preferred than $MO$ at $0$. The slope for the decrease gets steeper when $K$, the market depth, increases. This is because the affect of market order size gets smaller as $K$ increases.

\par The change of $m^*$ slope as $\beta$ can be found in Figure ~\ref{mystar_againstMu_diffBetas}. When $\beta$ is smaller, the slope is sharper. This is because when $\beta$ gets bigger, the market order size is more affective to the expected cash flow, so the difference in the market order size should be smaller compared to the case of a smaller $\beta$. 

\par The behavior of  $m^*$ for various values of $\sigma$ is shown in Figure~\ref{mystar_againstMu_diffSig_BM}.  As $\sigma$ increases, it is observed that $m^*$ increases. As described earlier, bigger $\sigma$ implies a higher chance of the limit order execution. However it is not necessarily true for a high $\mu$, and this explains why $y^*$ decreases and $m^*$ increases as $\sigma$ increases for a big enough $\mu$.

%
%
%

}

{
While in this paper we didn't include theorems for the behavior of the optimal placement strategy using a geometric Brownian motion, we have checked that they're showing a similar behavior as in a Brownian motion model.  
Bottom two panels of Figure~\ref{mystar_againstMu_diffRhos} show $m^*$, $y^*$ which maximize the ECF given in Lemma~\ref{ECF_GBM}. Since for the same price process, $\mu$ and $\sigma$ for the geometric Brownian motion model should be smaller than the those for BM, we are using $0.01$ for $\sigma$ and range of $\mu$ from $-0.1$ to $0.2$ for  the geometric Brownian motion model here. 
Comparing Brownian motion model (middle panels  of Figure~\ref{mystar_againstMu_diffRhos}) and geometric Brownian motion model (bottom panels of Figure~\ref{mystar_againstMu_diffRhos}) using same parameters (except $\mu$ and $\sigma$), $y^*$ and $m^*$ show similar behavior but it is smoother for the geometric Brownian motion model. 
}

{\Red
To summarize, we have investigated the  behavior of optimal ($y^*$, $m^*$) against important LOB features ($\rho$, $\beta$, $K$) and other parameters for price movement such as $\mu$ in a single period model.  In the following section, we will extend this investigation into a multiple period model. 
}

%

%
%
%

\newpage
\section{A Multi-Period Model}

\par In this section, we extend the single-period case to a multi-period case. In the single period case, the investor made a decision only at time $t=0$ and waited until the investor's time horizon $T$. As a natural extension, now we introduce $n$ time steps between time $0$ and time $T$. We find the optimal placement strategy at the first step. The difference is that all un-executed orders are not transferred to a market order at the end of the first period. Instead, now we solve a new optimal replacement problem using the remaining order. In other word, the remaining shares after the first period becomes a new $M$ for the second period, and so on.

\par Let us denote the time steps as $\{ t_0=0, t_1=T/n, t_2=2T/n, \dots, t_n=T \}$. At each $t_i$, the investor is allowed to 

\begin{itemize}
\item cancel any remaining limit order, (this becomes new remaining inventory, $M_i$)
\item place a market order of size $m_i$($\leq M_i$), or
\item place the rest ($M_i-m_i$) using  limit order at the price of $S_i+y_i$. ($S_i$: mid-price at time $t_i$)
\end{itemize}

Then we define the Multi-period Expected cash flow (MECF) as follows:

\par For $n=1$, 
\begin{align}\label{multi:mecf1:n1}
MECF(n=1,S_0,y,m,M,T=t)=ECF(y,m),
\end{align}

where the ECF (expected case flow) from the single-period case is in ~(\ref{ECF_lem3}).

\par For $n\geq 2$, we define the expected cash flow as a function of $n$, the number of time steps, $S_0$, the stock mid-price at time $t_0$, $y_0$, the initial limit order placement, $m_0$, the initial market order quantity, $M$, the inventory, and $T$, the time horizon.

{ Note that the cash flow at a time step $t_i$ depends only on $M_i, S_i$, and the action at $t_i$. Given this Markov structure, the optimal placement problem is a Markov decision problem where the expected cost for taking each action at each step can be solved recursively.

}

\begin{align}
\label{multi:mecf1:line1}
MECF (n, S_0,y_0, m_0, & M,T)=m_0\left(S(0,-m_0)-f  \right) 
+ E\left[L(S(0,0)+y+r)\bigg|\tau<\frac{T}{n} \right] P\left(\tau<\frac{T}{n}\right)
\\
\label{multi:mecf1:line2}
  &+MECF\bigg(n-1,E[{S_{\frac{T}{n}}\big|_{\tau<\frac{T}{n}}}],y_1,m_1,M-m_0-L, \frac{T(n-1)}{n}\bigg)P\left(\tau<\frac{T}{n}\right)\\
  \label{multi:mecf1:line3}
   &+
 MECF\bigg(n-1,E[{S_{\frac{T}{n}}\big|_{\tau>\frac{T}{n}}}],y_1,m_1,M-m_0, \frac{T(n-1)}{n}\bigg)P\left(\tau>\frac{T}{n}\right)
\end{align}

The first term of (\ref{multi:mecf1:line1}) is the cash flow from the market order, and the second term of (\ref{multi:mecf1:line1})
is the cash flow from executed limit order (recall that $L$ is the amount of executed shares at the price level $S_0+y_0$ when $\tau$, the first time the trader's limit order becomes the best ask, is less than $T/n$). Then, if $\tau<T/n$, the $M-m_0-L$ is the remaining inventory for the next time step, as described in (\ref{multi:mecf1:line2}). Otherwise, if $\tau>T/n$, then none of the limit order has been executed until $T/n$, and the remaining inventory for the next time step is $M-m_0$, as in (\ref{multi:mecf1:line3}).

\par For the dynamic placement problem, at each time step, the investor may update parameters and solve for the optimal placement. In this paper, we focus on solving for $y_0$ and $m_0$ assuming that the parameter remains the same during $[0,T]$. At each time step $t_i$, after updating the parameter, the investor can use the same technique to find the optimal  $y_i$ and $m_i$. 



\begin{prop}\label{multi:thm:rho}
Let MECF be defined as in (\ref{multi:mecf1:n1}),  (\ref{multi:mecf1:line1}) -(\ref{multi:mecf1:line3}). Let us assume that the price process, $S_t$, follows Brownian motion as in Lemma~\ref{lem:bm:ECF} and $L=(M-m)\rho$. Then, MECF(n, S, y, m, M, T) can be summarized as follows:

For $n=1$, 
\begin{align*}
MECF&(n=1,S_0,y,m,M,T=t)\\
=
&M \left(S_0 -2\rho(y-d)\tilde{N}(-\alpha_t)-(d+f-\mu t)(1-\rho+\rho\epsilon_0)+(1-\epsilon_0)\rho(y+r)\right)\\
&-m\left( \rho(1-\epsilon_0)(d+f+y+r)-2\rho (y-d)\tilde{N}(-\alpha_t) +\mu t(1-\rho+\rho\epsilon_0) \right)\\
&-\beta\left[m (K-m)^{-}+(M-m)\left(\epsilon_0  (K-M+m)^{-}+(1-\epsilon_0)(1-\rho)(K-(M-m)(1-\rho))^{-}\right) \right],\\
\end{align*}

\[
\alpha_t= \frac{y-d+\mu t}{\sigma \sqrt{t}},
\quad
\beta_t = \frac{y-d-\mu t}{\sigma\sqrt{t}}.
\quad \tilde{N}(-\alpha_t)= e^{2(y-d)\mu/\sigma^2}N(-\alpha_t),\quad 
\epsilon_0= \left(N\left(\beta_t \right)-\tilde{N}\left(-\alpha_t\right) \right)
\]

For $n\geq 2$,

\begin{gather}
	\begin{aligned}
MECF &(n, S_0,y_0^{n}, m_0^{n},M,T)\\
=&m_0^{n}(S_0-d+\beta \min (K-m_0^{n},0)-f)\\
+&(M-m_0^{n})\rho \left\{ (S_0+y_0^{n}+r) \left(N\left(-\beta_{\frac{T}{n}} \right)+ \tilde{N}\left(-\alpha_{\frac{T}{n}}\right) \right)\right\}\\
+&MECF(n-1,E[{S_{\frac{T}{n}}|_{\tau<T/n}}],y_1^{n},m_1^{n},(M-m_0^{n})(1-\rho)
, \frac{T(n-1)}{n} ) \left(N\left(-\beta_{\frac{T}{n}} \right)+\tilde{N}\left(-\alpha_{\frac{T}{n}}\right) \right)\\
+&MECF(n-1,E[{S_{\frac{T}{n}}|_{\tau>T/n}}],y_1^{n},m_1^{n},M-m_0^{n}, \frac{T(n-1)}{n}) \left(N\left(\beta_{\frac{T}{n}} \right)-\tilde{N}\left(-\alpha_{\frac{T}{n}}\right) \right)
\end{aligned}\label{multi:rho1:n}
\end{gather}

\[
\alpha_t= \frac{y-d+\mu t}{\sigma \sqrt{t}},
\quad
\beta_t = \frac{y-d-\mu t}{\sigma\sqrt{t}}.
\quad \tilde{N}(-\alpha_t)= e^{2(y-d)\mu/\sigma^2}N(-\alpha_t).
\]
\end{prop}

\begin{proof}
The result of this theorem is directly using the definition in  (\ref{multi:mecf1:n1}),  (\ref{multi:mecf1:line1}) -(\ref{multi:mecf1:line3}) and Lemma~\ref{BM_Lem}. 
\end{proof}

In this theorem we introduced the expression of the Expected Cash flow for the multiple period. Now, let us define the optimal placement for the first time step, $y_0^{*n}, m_0^{*n}$,

\begin{equation}\label{opt:ym:multi}
y_0^{*n}, m_0^{*n} := \underset{y\geq d,0\leq m\leq M}{\arg\max} MECF(n,S_0,y,m,T),
\end{equation}

 the optimal price level for the limit order placement (the limit order will be placed at $S_0+y_0^{*n}$) and the market order quantity $m_0^{*n}(\leq M)$, when there are $n$ time steps and $T$ is the total investment horizon.

 {
 Note that the action at each time step $t_i$, $y_i^{n}$ and $m_i^{n}$ depends on the state of $t_i$ (remaining inventory and the mid-price at $t_i$), so the optimal policy at time $t_i$, $(y_i^{*n}, m_i^{*n})$, could be degenerated into $\left((y_i^{*n}, m_i^{*n}), (y_{i+1}^{*n}, m_{i+1}^{*n}), \dots (y_{n-1}^{*n}, m_{n-1}^{*n})\right)$. Then, the Dynamic Programming Principle leads to the following backward recursion for the MECF:
 }
%

\begin{align*}
(y_0^{*n},& m_0^{*n})\\
=&\underset{y\geq d,0\leq m\leq M}{\arg\max} \bigg[ m(S_0-d+\beta \min (K-m,0)-f)+ (M-m)\rho  (S_0+y+r) \left(N\left(-\beta_{\frac{T}{n}} \right)+\tilde{N}\left(-\alpha_{\frac{T}{n}}\right) \right)\\
\quad &+\max_{y_1,m_1} MECF(n-1,E[{S_{\frac{T}{n}}|_{\tau<T/n}}],y_1,m_1,(M-m)(1-\rho)
, \frac{T(n-1)}{n} ) \left(N\left(-\beta_{\frac{T}{n}} \right)+\tilde{N}\left(-\alpha_{\frac{T}{n}}\right) \right)\\
&+ \max_{y_1, m_1}MECF(n-1,E[{S_{\frac{T}{n}}|_{\tau>T/n}}],y_1,m_1,M-m, \frac{T(n-1)}{n} ) \left(N\left(\beta_{\frac{T}{n}} \right)-\tilde{N}\left(-\alpha_{\frac{T}{n}}\right) \right)\bigg].
\end{align*}

%
\par Using this Proposition~\ref{multi:thm:rho}, we are now going to solve the optimal placement decision at time $0$. Because of the recursive nature of the multiple period, we are interested in analyzing the solution of the first period. Optimal behaviors for the remaining period are updated at next time steps with updated information (parameters), using the same method, sequentially. We now provide analytical results for the optimal behavior. The following theorem gives analytical result for the case when $\rho=1$.

\begin{thm}\label{Multi:rho1:optimalym}
For $y_0^{*n}, m_0^{*n}$ defined as in Proposition~\ref{multi:thm:rho} and (\ref{opt:ym:multi}) and when $\rho=1$,  $y_0^{*n}=d, m_0^{*n}=0$ for all $n$ when $\mu<0$. 
\end{thm}

\begin{proof}

For $n=1$, proof is done in Proposition~\ref{thm:single:muNeg}.
Note that $y_0^{*1}=d$, $m_0^{*1}=0$,  $\max MECF= MECF(n=1, S_0, y=d,m=0,M,T)=M(S_0+d+r)$. 

Now, let us prove that for $n\geq 2$, if
$y_0^{*(n-1)}=d, m_0^{*(n-1)}=0$ and  $\max MECF(n-1,S_0,y,m,M,T)=M (S_0 + d + r)$, then
$y_0^{*n}=d, m_0^{*n}=0$ and  $\max MECF(n,S_0,y,m,M,T)=M (S_0 + d + r)$. 

%
%
%
To prove this, let us denote 
\begin{align*}
	f(y,m):=&m(S_0-d+\beta \min (K-m,0)-f)+(M-m)(S_0+y+r) P(\tau<T/n)\\
&+(M-m) E[{S_{\frac{T}{n}}I({\tau>T/n})}] + (M-m)(d+r) P ({\tau>T/n}).
\end{align*}

 Then, from Eq.~(\ref{multi:rho1:n}), $\max_{y,m} MECF (n, S_0,y, m,M,T)=\max_{y,m} f(y,m)$. We now need to find $y,m$ which maximize $f(y,m)$. 
%

Let us recall that $t_1:=T/n$. Using Lemma~\ref{BM_Lem}, $f(y,m)$ can be written as 
\begin{align*}
	 f(y,m)=&MS_0+m(-d+\beta \min (K-m,0)-f) +(M-m)(d+r)\\
&+(M-m)(y-d) (N(-\beta_{t_1})+\tilde{N}(-\alpha_{t_1} ))
+(M-m) \left(\mu {t_1} N(\beta_{t_1}) + (-2(y-d)-\mu {t_1} ) \tilde{N} (-\alpha_{t_1}) \right), 
\end{align*}

And when we take a derivative with respect to $y$, 
$
	{\partial f(y,m)}/{\partial y} =(M-m) \big( N(-\beta_{t_1})-\tilde{N}(-\alpha_{t_1}) -\frac{2\mu^2 t}{\sigma^2} N(-\alpha_{t_1}) \big) <0,
$
so $y^*=d$. 

Now, $f(y=d,m)= M(S_0+d+r) +m(-2d+\beta \min (K-m,0)-f-r) $, which is linearly decreasing when $m\leq K$, and when $m>K$, $f(y=d, m)$ becomes a second order concave polynomial of $m$ which takes its maximum at $K/2-(2d+f+r)/\beta<K$, so for $m>K$, $f(y=d, m)$ is also a decreasing function of $m$. Therefore, $f(y=d, m)$ is a decreasing function of $m$ for $0\leq m \leq M$, so $m^*=0$. Also, $\max MECF(n,S_0,y,m,M,T)=MECF(n,S_0,y=d,m=0,M,T)= M(S_0+d+r).$

We have proved that that if $y_0^{*(n-1)}=d, m_0^{*(n-1)}=0$ and  $\max MECF(n-1,S_0,y,m,M,T)=M (S_0 + d + r)$, then
$y_0^{*n}=d, m_0^{*n}=0$ and  $\max MECF(n,S_0,y,m,M,T)=M (S_0 + d + r)$. 
By induction, this proves the theorem.
\end{proof}

{
In Theorem~\ref{Multi:rho1:optimalym}, we've shown the analysis for the case when $\rho=1$ and $\mu<0$, which means that the limit order execution is guaranteed when the price process reaches the price level of the limit order, and the drift is negative. This could represent the liquid market condition. In this case, the optimal strategy is to place all using the limit order at the best ask, and this does not change for the number of time steps. 

\par Now, let us show that $y^{*2}=d$ for any $\rho \in [0,1]$ for the negative drift. Due to the complexity, we've proved for the case when $n=2$.
}

\begin{thm}\label{Multi:rho:optimaly}
For $y_0^{*n}$ defined as in Proposition~\ref{multi:thm:rho} and (\ref{opt:ym:multi}),  $y_0^{*2}=d$  when $\mu<0$. 
\end{thm}

\begin{proof}
For $n=1$, recall that

\begin{gather}\nonumber
	\begin{aligned}
	ECF(y=d,m,M)=&M (S_0+\rho(d+r)-(1-\rho)(d+f-\mu t))-m\epsilon_3 +m\beta \min (K-m,0)\\
	&+\beta (M-m)(1-\rho) \min(K-(M-m)(1-\rho),0),
\end{aligned}
\end{gather}

Let us denote 

$$ 
F(M,m):=ECF(y=d,m,M)-M S_0 .
$$

where $\epsilon_3:=\rho(2d+r+f)+\mu t(1-\rho)$.

For $n=2$, let $M_A:=(M-m_0)(1-\rho)$ and $M_B:= (M-m_0)$. 
MECF can be written as

\begin{gather}
	\begin{aligned}
MECF &(n=2, S_0,y, m_0,M,T)\\
=&M S_0 +  m_0(-d+\beta \min (K-m_0,0)-f)\\
&+(M-m_0)\rho \left\{ (y+r) \left(N\left(-\beta_{\frac{T}{2}} \right)+ \tilde{N}\left(-\alpha_{\frac{T}{2}}\right) \right)\right\}\\
&+
E[{S_{\frac{T}{2}}\mathcal{I}({\tau<T/2})}](M-m_0)(1-\rho) + \max_{0\leq m\leq M_A}F( M_A,m)\left(N\left(-\beta_{\frac{T}{2}} \right)+\tilde{N}\left(-\alpha_{\frac{T}{2}}\right) \right)\\
&+E[{S_{\frac{T}{2}}\mathcal{I}({\tau>T/2})}](M-m_0) + \max_{0\leq m\leq M_B}F( M_B,m)\left(N\left(\beta_{\frac{T}{2}} \right)-\tilde{N}\left(-\alpha_{\frac{T}{2}}\right) \right).\\
\end{aligned}\label{multi:rho1:n}
\end{gather}

By taking derivative with respect to $y$,

\begin{align*}
	\frac{\partial }{\partial y}& MECF (n=2, S_0,y, m_0,M,T) \\ 
	=&\rho (M-m_0) \left(N\left(-\beta_{\frac{T}{2}} \right)-\tilde{N}\left(-\alpha_{\frac{T}{2}}\right) \right)\\
	&+\frac{2e^{\frac{2(y-d)\mu}{\sigma^2}}}{\sigma\sqrt{{\frac{T}{2}}}}\left( \phi (\alpha_{\frac{T}{2}} ) -\alpha_{\frac{T}{2}} N(-\alpha_{\frac{T}{2}})\right)\left(-\rho(r+d-\mu{\frac{T}{2}})(M-m_0)+ \max_{0\leq m\leq M_B}F( M_B,m)- \max_{0\leq m\leq M_A}F( M_A,m)\right)\\
	&+\frac{2(y-d)\tilde{N}\left(-\alpha_{\frac{T}{n}}\right)}{\sigma^2 {\frac{T}{2}}}\left(-\rho(r+d)(M-m_0)+ \max_{0\leq m\leq M_B}F( M_B,m)- \max_{0\leq m\leq M_A}F( M_A,m)\right)
\end{align*}

Since first line is negative, if we show that $\left(\max_{0\leq m\leq M_B}F( M_B,m)- \max_{0\leq m\leq M_A}F( M_A,m)\right)<\rho(r+d)(M-m_0)$, we can show that the entire expression is negative. The proof follows:

\begin{gather}
	\begin{aligned}
			&\max_{0\leq m\leq M_B}F( M_B,m)- \max_{0\leq m\leq M_A}F( M_A,m)\\
	=& \rho M_B (\rho (d+r)-(1-\rho) (d+f-\mu {\frac{T}{2}}))\\
	&+ \max_{0\leq m\leq M_B}\left[
	-m\epsilon_3 +m\beta \min (K-m,0)+\beta (M_B-m)(1-\rho) \min(K-(M_B-m)(1-\rho),0)
	\right]
	\\
	&- \max_{0\leq m\leq M_A}\left[
	-m\epsilon_3 +m\beta \min (K-m,0)+\beta (M_A-m)(1-\rho) \min(K-(M_A-m)(1-\rho),0)
	\right]
	\end{aligned}\label{thm5:proof:eq1}
\end{gather}

Let's denote $m^{b}:=\arg\max F( M_B,m)$. Let's divide into two cases:  $m^{b}\leq M_A=M_B(1-\rho)$ and $ M_A=M_B(1-\rho)<m^{b}<M_B$.

\noindent\textbf{Case 1,  $m^{b}\leq M_A=M_B(1-\rho)$:}
In this case, $\max_{0\leq m\leq M_A}F( M_A,m)\geq F(M_A,m=m^b)$, so
$\max_{0\leq m\leq M_B}F( M_B,m)- \max_{0\leq m\leq M_A}F( M_A,m) \leq 
F(M_B,m=m^b)-F(M_A,m=m^b)$, therefore


\begin{align*}
	&(\ref{thm5:proof:eq1})\leq \rho M_B (\rho (d+r)-(1-\rho) (d+f-\mu {\frac{T}{2}}))\\
	&+ \left[
	-m^{b}\epsilon_3 +m^{b}\beta \min (K-m^{b},0)+\beta (M_B-m^{b})(1-\rho) \min(K-(M_B-m^{b})(1-\rho),0)
	\right]
	\\
	&+ \left[
	-m^{b}\epsilon_3 +m^{b}\beta \min (K-m^{b},0)+\beta (M_A-m^{b})(1-\rho) \min(K-(M_A-m^{b})(1-\rho),0)
	\right]\\
&= \rho^2 M_B(d+r) - \rho (1-\rho) (d+f-\mu {\frac{T}{2}}) M_B\\
&+\beta (M_B-m^{b})(1-\rho) \min(K-(M_B-m^{b})(1-\rho),0)-\beta (M_A-m^{b})(1-\rho) \min(K-(M_A-m^{b})(1-\rho),0)
\\&< \rho^2 M_B(d+r)\leq (M-m_0)\rho(d+r).
\end{align*}
 
%
%

\noindent\textbf{Case 2, $ M_A=M_B(1-\rho)<m^{b}<M_B$:}

In this case, $\max_{0\leq m\leq M_A}F( M_A,m)\geq F(M_A,m=M_A)$, so
$\max_{0\leq m\leq M_B}F( M_B,m)- \max_{0\leq m\leq M_A}F( M_A,m) \leq 
F(M_B,m=m^b)-F(M_A,m=m_A)$, therefore
\begin{align*}
	&(\ref{thm5:proof:eq1})\leq \rho M_B (\rho (d+r)-(1-\rho) (d+f-\mu{\frac{T}{2}}))\\
	&+ \left[
	-m^{b}\epsilon_3 +m^{b}\beta \min (K-m^{b},0)+\beta (M_B-m^{b})(1-\rho) \min(K-(M_B-m^{b})(1-\rho),0)
	\right]
	\\
	&- \left[
	-M_A\epsilon_3 +M_A\beta \min (K-M_A,0)
	\right]\\
&= \rho^2 M_B(d+r) - \rho (1-\rho) (d+f-\mu {\frac{T}{2}}) M_B\\
&+ \bigg[
	-m^{b}\left( \rho(2d+r+f) + \mu {\frac{T}{2}} (1-\rho) -\beta \min (K-m^{b},0) \right)\\
	 &\quad +\beta (M_B-m^{b})(1-\rho) \min(K-(M_B-m^{b})(1-\rho),0)
	\bigg]
	\\
	&- \left[-M_A\rho(2d+r+f) -M_A \mu {\frac{T}{2}} (1-\rho) +M_A\beta \min (K-M_A,0)
	\right]\\
	&<\rho^2 M_B(d+r)+\rho(-1\rho)\mu {\frac{T}{2}} M_B -m^{b} \mu t (1-\rho) +M_A\mu {\frac{T}{2}}(1-\rho)<\rho^2 M_B(d+r)\leq \rho (M-m_0)(d+r). 
\end{align*}

So we proved that  $\left(\max_{0\leq m\leq M_B}F( M_B,m)- \max_{0\leq m\leq M_A}F( M_A,m)\right)<\rho(r+d)(M-m_0)$, which implies  $\frac{\partial }{\partial y} MECF (n, S_0,y, m_0,M,T)<0$. Therefore,  $y^*=d$.

\end{proof}

\begin{remk}
While we've only proved for the case for two-period case, $y_0^{*n}=d$ appears to be true for all $n$. 
\end{remk}

Now, let us study what happens to $m_{0}^{*n}$. Let us show the behavior of $m_0^{*2}$. 

\begin{thm}
	For $m_0^{*n}$ defined as in Proposition~\ref{multi:thm:rho},  $m_0^{*2}\leq m_0^{*1}$  when $\mu<0$.
	For $m_0^{*2}$,

\begin{enumerate}
	\item If $\rho=1$, $m_0^{*2}=0$.
	\item If $\rho<1$ and $K\geq M^{`}$,
		\begin{itemize}
			\item $m_0^{*2}=0$ when $\rho(3d+2r+f)+\mu t(1-\rho)>0$
			\item $m_0^{*2}=M^{`}$ when $\rho(3d+2r+f)+\mu t(1-\rho)<0$
		\end{itemize}
	\item If $\rho<1$ and $K< M^{`}$,	
	\item[] $m_0^{*2}$ is one of $\{0,K, M^{`}-K/(1-\rho), M^{`}, m^{*2}_1, m^{*2}_2, m^{*3}_3\}$ where  
\end{enumerate}

$$ 
m^{*2}_1=K/2-\epsilon_3^{`}/2\beta, \quad
m^{*2}_2=M^{`}-\frac{K}{2(1-\rho)}-\frac{\epsilon_3^{`}}{2(1-\rho)^2\beta},
\quad 
m^{*2}_3=\frac{2M^{`}\beta (1-\rho)^2 +\beta\rho K-\epsilon_3^{`}}{2\beta(1+(1-\rho)^2)},
$$ 

where
 $\epsilon_3^{`}:=\rho(3d+2r+f)+(1-\rho)\mu T$ and $M^{`}:=M-m_1$, 

$m_1:=\arg\max_{0\leq m\leq M} ECF(n=1,S_0,y=d,m, M, T/2)$, which can be computed using Theorem~\ref{thm:single:muNeg}.

\end{thm}

\begin{proof}
Note that the MECF for $n=2$ can be written as following, where the notations are borrowed from the proof of Theorem~\ref{Multi:rho1:optimalym}.
\begin{gather}
	\begin{aligned}
MECF &(n=2, S_0,y=d, m_0,M,T)\\
=&M S_0 +  m_0(-d+\beta \min (K-m_0,0)-f)
+(M-m_0)\rho \left\{ (d+r)\right\}\\
+&\mu \frac{T}{2}(M-m_0)(1-\rho) + (M-m_0)(\rho(d+r)-(1-\rho)(d+f-\mu T/2))\\
 + &\max_{0\leq {m_1}\leq M_A}[-{m_1}\epsilon_3 +{m_1}\beta \min (K-m,0)
	+\beta (M_A-{m_1})(1-\rho) \min(K-(M_A-{m_1})(1-\rho),0)],
\end{aligned}
\end{gather}

$M_A=(M-m_0)(1-\rho)$.

$m^{}$ which minimizes $MECF (n=2, S_0,y=d, m,M,T)$ actually minimizes the following with respect to $m_0$:

\begin{align}\nonumber
&-m_0 \beta min(K-m_0,0) -m_0 \epsilon_3^{`} + \beta (M^{`}-m_0 )(1-\rho) min (K- (M^{'}-m_0  )(1-\rho),0), \nonumber
\end{align}

while the solution for $n=1$ minimizes 
$$
-m \beta min(K-m,0) -m \epsilon_3 + \beta (M-m)(1-\rho) min (K- (M-m)(1-\rho),0),
$$
with respect to m. Note that $\epsilon_3:=\rho(2d+r+f)+\mu t(1-\rho)$, $\epsilon_3^{`}:=\rho(3d+2r+f)+(1-\rho)\mu T$ and $M^{`}:=M-m_1$.

Using the proof of Theorem~\ref{thm:single:muNeg}. we have

\begin{enumerate}
	\item If $\rho=1$, $m_0^{*2}=0$.
	\item If $\rho<1$ and $K\geq M^{`}$,
		\begin{itemize}
			\item $m_0^{*2}=0$ when $\rho(3d+2r+f)+\mu t(1-\rho)>0$
			\item $m_0^{*2}=M^{`}$ when $\rho(3d+2r+f)+\mu t(1-\rho)<0$
		\end{itemize}
	\item If $\rho<1$ and $K< M^{`}$,	
	\item[] $m_0^{*2}$ is one of $\{0,K, M^{`}-K/(1-\rho), M, m^{*2}_1, m^{*2}_2, m^{*3}_3\}$ where  
\end{enumerate}

$$ 
m^{*2}_1=K/2-\epsilon_3^{`}/2\beta, \quad
m^{*2}_2=M^{`}-\frac{K}{2(1-\rho)}-\frac{\epsilon_3^{`}}{2(1-\rho)^2\beta},
\quad 
m^{*2}_3=\frac{2M^{`}\beta (1-\rho)^2 +\beta\rho K-\epsilon_3^{`}}{2\beta(1+(1-\rho)^2)}.
$$ 

\begin{remk}
Note that $\epsilon_3<\epsilon_3^{`}$, and $M^{`}<M$, which makes $m^{*2}_1\leq m_1^*$,  $m^{*2}_2 \leq m_2^*$,  $m^{*2}_3\leq m_3^*$. From this inequality, it is observed that $m^{*2}_0$, the optimal Market Order quantity at time $0$ for the two-step case is less than (or equal to ) $m^*_0$, the optimal Market Order quantity at time $0$ for the one-step case. 

\end{remk}

\end{proof}

\begin{figure}
	\includegraphics[width=0.495\textwidth]{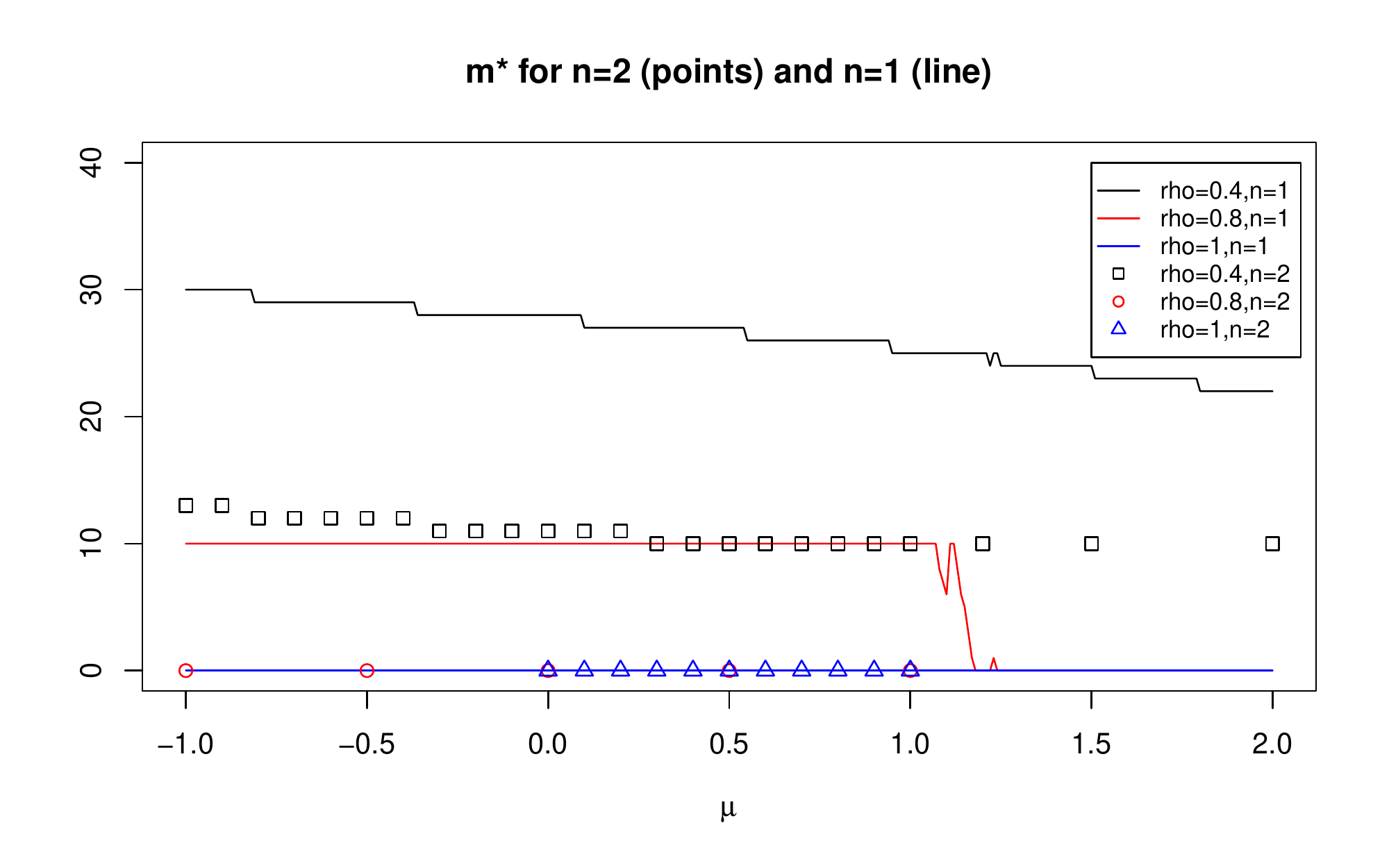}
	\includegraphics[width=0.495\textwidth]{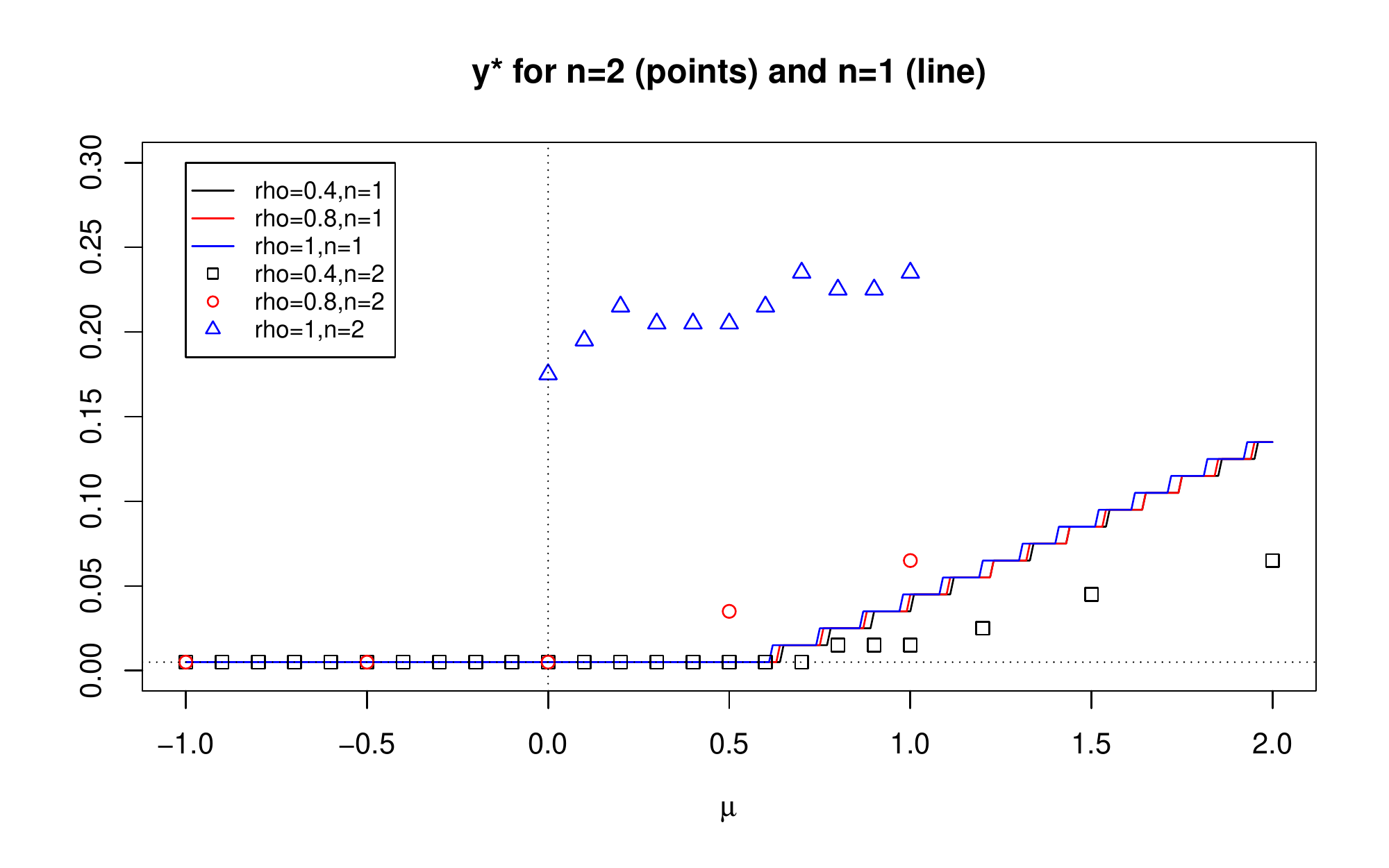}
	\caption{$m^*$ (left) and $y^*$ (right) for multi-period (dots) and single period (lines).  $M=100, r=0.003, f=0.003, \epsilon=0.01, d=\epsilon/2=0.005, T=0.1, S_0=100, \sigma=0.1$. 
	}
	\label{multigraph}
\end{figure}

{

Behaviors of $y_0^{*1}$ (line) and $y_0^{*2}$ (points) for different $\rho$ values (black: $\rho=0.4$, red: $\rho=0.8$, blue: $\rho=1$ are described in  the right panel of Figure~\ref{multigraph}. We observe that when $\mu<0$, $y_0^{*2}=d$ as described in Theorem~\ref{Multi:rho:optimaly}. 
 For $\mu>0$, it appears that it still increases in a linear (step-wise linear) way, but its slope against $\rho$ is more sensitive to the value of $\rho$ for $n=2$. More specifically, when $\rho=0.4$, $y^{*2}<y^{*1}$, but when $\rho=0.8, 1$, $y^{*2}>y^{*1}$. This is explained by the chance of the $\max ECF$ for the single period. 

For the single period, when limit orders are not executed, the investor use a market order at time $T$. Then, the expected stock price does not change with respect to $\rho$, and the change of $y^*$ w.r.t. $\rho$ in the single period is minimal. 

However,  the $\max ECF$ (for the single period) changes when $\rho$ changes. 
When $n=2$, if the limit order placed at $t=0$ is not executed, then the remaining orders are placed following the algorithm to maximize the ECF of the single period, and that maximum changes with respect to $\rho$. Therefore,  $y^{*2}$ changes more sensitively to the change of $\rho$ than $y^*$ . We observe that when the number of time step increases, the slope change will more various.

The left panel of Figure~\ref{multigraph} shows behaviors of $m_0^{*1}$ and $m_0^{*2}$ against $\mu$ for different $\rho$s. 
It appears that $m_0^{*2}\leq m_0^{*1}$, which implies that the optimal quantity for the market order at time $0$ decreases when $n$, number of time step, increases. This is because as n increases, it is better to divide the market order into multiple chunks to reduce the negative impact from the size of the market order. Existence of $K$ suggests that placing a  smaller amount of a market order twice is better than placing a big amount of at once even though the total MO amount is same.

}

\section{Conclusion}
\par Throughout this paper, we have obtained explicit solutions and approximations for the optimal execution problem under the liquidity cost. First, we have derived the explicit solution for the optimal size of the market order, $m^*$, for the single order. In addition, the approximation of the optimal price level for the limit order, $y^*$, has been calculated. The behavior of $y^*$ and $m^*$ have been investigated under various market conditions.  Finally, we have extended this problem into multiple-period problem where investors can change their decision at multiple time steps. We have shown the change of the optimal behavior for different number of time steps.

\par {\Red
To summarize, we have investigated the  behavior of optimal ($y^*$, $m^*$) against important LOB features and other parameters for price movement in single and multiple period models. Instead of modeling the whole LOB movement, we have taken essential features of LOB to simplify the optimization without losing the important traits of LOB. 
}

\par There are other important features in the Limit Order Book closely related to this paper which have not been studied yet. For instance, the consideration of movement of bid-ask spread, the estimation and the update of parameters, the effect from correlated assets, the realistic diffusive models with jumps, are all worth further development. The models in this study could enlighten such developments.

\bibliography{thebib_LiqRisk_19}

\begin{thebibliography}{24}
\providecommand{\natexlab}[1]{#1}
\providecommand{\url}[1]{\texttt{#1}}
\expandafter\ifx\csname urlstyle\endcsname\relax
  \providecommand{\doi}[1]{doi: #1}\else
  \providecommand{\doi}{doi: \begingroup \urlstyle{rm}\Url}\fi

\bibitem[Alfonsi et~al.(2010)Alfonsi, Fruth, and Schied]{alfonsi2010}
A.~Alfonsi, A.~Fruth, and A.~Schied.
\newblock Optimal execution strategies in limit order books with general shape
  functions.
\newblock \emph{Quantitative Finance}, 10\penalty0 (2):\penalty0 143--157,
  2010.

\bibitem[Alfonsi et~al.(2012)Alfonsi, Schied, and Slynko]{alfonsi2012order}
A.~Alfonsi, A.~Schied, and A.~Slynko.
\newblock Order book resilience, price manipulation, and the positive portfolio
  problem.
\newblock \emph{SIAM Journal on Financial Mathematics}, 3\penalty0
  (1):\penalty0 511--533, 2012.

\bibitem[Almgren and Chriss(2001)]{almgren2001optimal}
R.~Almgren and N.~Chriss.
\newblock Optimal execution of portfolio transactions.
\newblock \emph{Journal of Risk}, 3:\penalty0 5--40, 2001.

\bibitem[Bank and Baum(2004)]{bank2004hedging}
P.~Bank and D.~Baum.
\newblock Hedging and portfolio optimization in financial markets with a large
  trader.
\newblock \emph{Mathematical Finance: An International Journal of Mathematics,
  Statistics and Financial Economics}, 14\penalty0 (1):\penalty0 1--18, 2004.

\bibitem[Bertsimas and Lo(1998)]{bertsimas1998optimal}
D.~Bertsimas and A.~W. Lo.
\newblock Optimal control of execution costs.
\newblock \emph{Journal of Financial Markets}, 1\penalty0 (1):\penalty0 1--50,
  1998.

\bibitem[Blais and Protter(2010)]{BP:2010}
M.~Blais and P.~Protter.
\newblock An analysis of the supply curve for liquidity risk through book data.
\newblock \emph{International Journal of Theoretical and Applied Finance},
  13\penalty0 (06):\penalty0 821--838, 2010.

\bibitem[Cartea et~al.(2014)Cartea, Jaimungal, and Ricci]{cartea2014modelling}
A.~Cartea, S.~Jaimungal, and J.~Ricci.
\newblock Buy low, sell high: a high frequency trading perspective.
\newblock \emph{SIAM J. Financ. Math.}, 5\penalty0 (1):\penalty0 415--444,
  2014.

\bibitem[\c{C}etin et~al.(2010)\c{C}etin, Jarrow, and Protter]{CJP:2003}
U.~\c{C}etin, R.~A. Jarrow, and P.~Protter.
\newblock Liquidity risk and arbitrage pricing theory.
\newblock In \emph{Handbook of Quantitative Finance and Risk Management}, pages
  1007--1024. Springer, 2010.

\bibitem[Cont and Kukanov(2013)]{ContKukanov2013}
R.~Cont and A.~Kukanov.
\newblock Optimal order placement in limit order markets.
\newblock \emph{Available at ssrn 2155218}, 2013.

\bibitem[Donier(2012)]{donier2012market}
J.~Donier.
\newblock Market impact with autocorrelated order flow under perfect
  competition.
\newblock \emph{Available at SSRN 2191660}, 2012.

\bibitem[Eisler et~al.(2012)Eisler, Bouchaud, and Kockelkoren]{eisler2012price}
Z.~Eisler, J.-P. Bouchaud, and J.~Kockelkoren.
\newblock The price impact of order book events: market orders, limit orders
  and cancellations.
\newblock \emph{Quantitative Finance}, 12\penalty0 (9):\penalty0 1395--1419,
  2012.

\bibitem[Figueroa-L{\'o}pez et~al.(2018)Figueroa-L{\'o}pez, Lee, and
  Pasupathy]{FLP:2018}
J.~E. Figueroa-L{\'o}pez, H.~Lee, and R.~Pasupathy.
\newblock Optimal placement of a small order in a diffusive limit order book.
\newblock \emph{High Frequency}, 1\penalty0 (2):\penalty0 87--116, 2018.

\bibitem[Frey and Patie(2002)]{frey2002risk}
R.~Frey and P.~Patie.
\newblock Risk management for derivatives in illiquid markets: A simulation
  study.
\newblock In \emph{Advances in finance and stochastics}, pages 137--159.
  Springer, 2002.

\bibitem[Gatheral(2010)]{gatheral2010no}
J.~Gatheral.
\newblock No-dynamic-arbitrage and market impact.
\newblock \emph{Quantitative finance}, 10\penalty0 (7):\penalty0 749--759,
  2010.

\bibitem[Gatheral et~al.(2012)Gatheral, Schied, and
  Slynko]{gatheral2012transient}
J.~Gatheral, A.~Schied, and A.~Slynko.
\newblock Transient linear price impact and fredholm integral equations.
\newblock \emph{Mathematical Finance: An International Journal of Mathematics,
  Statistics and Financial Economics}, 22\penalty0 (3):\penalty0 445--474,
  2012.

\bibitem[Gu{\'e}ant(2015)]{gueant2015optimal}
O.~Gu{\'e}ant.
\newblock Optimal execution and block trade pricing: a general framework.
\newblock \emph{Applied Mathematical Finance}, 22\penalty0 (4):\penalty0
  336--365, 2015.

\bibitem[Guilbaud and Pham(2013)]{GuilbaudPham2013}
F.~Guilbaud and H.~Pham.
\newblock Optimal high-frequency trading with limit and market orders.
\newblock \emph{Quantitative Finance}, 13\penalty0 (1):\penalty0 79--94, 2013.

\bibitem[Guo et~al.(2016)Guo, De~Larrard, and Ruan]{GDR:2017}
X.~Guo, A.~De~Larrard, and Z.~Ruan.
\newblock Optimal placement in a limit order book.
\newblock \emph{Mathematics and Financial Economics}, 2016.

\bibitem[Jacquier and Liu(2017)]{jacquierLiu2017}
A.~Jacquier and H.~Liu.
\newblock Optimal liquidation in a level-i limit order book for large-tick
  stocks.
\newblock \emph{Preprint. Available at arXiv:1701.01327 [q-fin.TR]}, 2017.

\bibitem[Jeanblanc et~al.(2009)Jeanblanc, Yor, and
  Chesney]{jeanblanc2009mathematical}
M.~Jeanblanc, M.~Yor, and M.~Chesney.
\newblock \emph{Mathematical methods for financial markets}.
\newblock Springer, 2009.

\bibitem[Maglaras et~al.(2015)Maglaras, Moallemi, and Zheng]{Manglaras2015}
C.~Maglaras, C.~Moallemi, and H.~Zheng.
\newblock Optimal execution in a limit order book and an associated
  microstructure market impact model.
\newblock \emph{Preprint available at ssrn 2610808}, 2015.

\bibitem[Obizhaeva and Wang(2013)]{obizhaeva2013optimal}
A.~A. Obizhaeva and J.~Wang.
\newblock Optimal trading strategy and supply/demand dynamics.
\newblock \emph{Journal of Financial Markets}, 16\penalty0 (1):\penalty0 1--32,
  2013.

\bibitem[Potters and Bouchaud(2003)]{potters2003more}
M.~Potters and J.-P. Bouchaud.
\newblock More statistical properties of order books and price impact.
\newblock \emph{Physica A: Statistical Mechanics and its Applications},
  324\penalty0 (1-2):\penalty0 133--140, 2003.

\bibitem[Predoiu et~al.(2011)Predoiu, Shaikhet, and Shreve]{predoiu2011optimal}
S.~Predoiu, G.~Shaikhet, and S.~Shreve.
\newblock Optimal execution in a general one-sided limit-order book.
\newblock \emph{SIAM Journal on Financial Mathematics}, 2\penalty0
  (1):\penalty0 183--212, 2011.

\end{thebibliography}

\end{document}